\documentclass[runningheads]{llncs}
\usepackage[utf8]{inputenc}

\usepackage[bookmarks,colorlinks,breaklinks]{hyperref}
\hypersetup{urlcolor=blue, colorlinks=true, citecolor=green!50!black, linkcolor=blue}
\usepackage[T1]{fontenc}
\usepackage{bold-extra}
\usepackage[american]{babel}
\usepackage{etoolbox}
\usepackage{varwidth}
\usepackage[normalem]{ulem}

\usepackage{graphicx}
\usepackage{tabularx}

\usepackage{amsmath, amssymb, cases}
\usepackage{theoremstyles}
\usepackage{thm-restate}
\usepackage{enumitem}
\usepackage{mdframed}
\usepackage{bbm}
\usepackage{bm}
\usepackage{mathtools}

\usepackage[ruled]{algorithm}
\usepackage[noend]{algpseudocode}

\usepackage{microtype}

\usepackage{mdframed}
\usepackage[dvipsnames]{xcolor}

\usepackage{tikz}
\usetikzlibrary{fit, arrows, shapes, matrix, chains, positioning, backgrounds, arrows.meta, tikzmark, decorations.pathreplacing, calc}

\usepackage{mathrsfs}
\usepackage{makecell, multirow}
\setcellgapes{3pt}

\tikzset{me/.style={to path={
\pgfextra{%
 \pgfmathsetmacro{\startf}{-(#1-1)/2}  
 \pgfmathsetmacro{\endf}{-\startf} 
 \pgfmathsetmacro{\stepf}{\startf+1}}
 \ifnum 1=#1 -- (\tikztotarget)  \else
     let \p{mid}=($(\tikztostart)!0.5!(\tikztotarget)$) 
         in
\foreach \i in {\startf,\stepf,...,\endf}
    {%
     (\tikztostart) .. controls ($ (\p{mid})!\i*6pt!90:(\tikztotarget) $) .. (\tikztotarget)
      }
      \fi   
     \tikztonodes
}}}

\definecolor{color1}{RGB}{128,0,0}     
\definecolor{color2}{RGB}{170,110,40}   
\definecolor{color3}{RGB}{128,128,0}    
\definecolor{color4}{RGB}{0,128,128}     
\definecolor{color5}{RGB}{0,0,128}       
\definecolor{color6}{RGB}{0,0,0}        
\definecolor{color7}{RGB}{230,25,75}      
\definecolor{color8}{RGB}{245,130,48}  
\definecolor{color9}{RGB}{255,225,25}  
\definecolor{color10}{RGB}{210,245,60}    
\definecolor{color11}{RGB}{60,180,75}    
\definecolor{color12}{RGB}{70,240,240}   
\definecolor{color13}{RGB}{0,130,200}    
\definecolor{color14}{RGB}{145,30,180}   
\definecolor{color15}{RGB}{240,50,230}   
\definecolor{color16}{RGB}{128,128,128}  
\definecolor{color17}{RGB}{250,190,212}  
\definecolor{color18}{RGB}{215,215,180}  
\definecolor{color19}{RGB}{255,250,200}  
\definecolor{color20}{RGB}{170,255,195}  
\definecolor{color21}{RGB}{220,190,255}  
\definecolor{color22}{RGB}{255,255,255}  

\newcommand{\disjointTSP}{\ensuremath{\texorpdfstring{\textsc{Disjoint-TSP}_2}{TSP2}}}
\newcommand{\disjointSHP}{\ensuremath{\texorpdfstring{\textsc{Disjoint-Path TSP}_2}{Path TSP2}}}
\newcommand{\uniformLine}{\ensuremath{\texorpdfstring{\mathcal{R}_n^1}{R1n}}}
\newcommand{\uniformLineGeomGraph}{\ensuremath{\texorpdfstring{R_n^1}{R1n}}}
\newcommand{\uniformCircle}{\ensuremath{\texorpdfstring{\mathcal{S}_n^1}{S1n}}}
\newcommand{\uniformCircleGeomGraph}{\ensuremath{\texorpdfstring{S_n^1}{S1n}}}
\newcommand{\realLine}{\ensuremath{\texorpdfstring{\mathbb{R}^1}{R1}}}
\newcommand{\realCircle}{\ensuremath{\texorpdfstring{\mathbb{S}^1}{S1}}}
\newcommand{\generalClassGeomGraphs}{\ensuremath{\texorpdfstring{\mathcal{M}}{M}}}

\begin{document}
\title{Disjoint Tours and the Price of Diversity\thanks{This research was supported by the European Union’s Horizon 2020 research and innovation programme under the Marie Skłodowska-Curie grant agreement no. 945045, and by the NWO Gravitation project NETWORKS under grant no. 024.002.003.}
}
\author{Mark de Berg \and
Andr\'es {L\'opez Mart\'inez} \and
Frits Spieksma}
\authorrunning{M. de Berg et al.}
\institute{Department of Mathematics and Computer Science\\TU Eindhoven}
\maketitle              
\begin{abstract}

We study a variant of the Traveling Salesman Problem, where instead of finding a single tour, we want to find a pair of two edge-disjoint tours whose longer tour is as short as possible. We investigate the \emph{Price of Diversity (PoD)} for this problem, which is the ratio of the cost of the longer of the two tours and the cost of a single optimal tour, in the worst case over all possible instances. We prove (almost) tight bounds on this quantity for a special 1-dimensional scenario and for general metric spaces. We believe that the Price-of-Diversity framework that we introduce is interesting in its own right, and may lead to follow-up work on other problems as well.

\keywords{Diverse Solutions \and TSP \and Path TSP \and Price of Diversity}
\end{abstract}
\section{Introduction}
\label{sec:intro} 

Traditionally, algorithms for combinatorial optimization problems focus on computing a single, optimal solution. There are many situations, however, where it is desirable to compute multiple solutions that are sufficiently distinct. For example, a decision maker who lacks complete information about the costs or feasibility of candidate solutions may prefer to keep multiple alternatives open until the missing details are revealed, reducing the risk of committing to a suboptimal choice. Likewise, in time-sensitive settings where last-minute constraints can arise, a diverse solution set increases the likelihood that at least one feasible option is available.

Typically it is not possible to find multiple distinct solutions that are all optimal. Thus, requiring a diverse set of solutions comes at a price. In this paper, we study this phenomenon for the Traveling Salesman Problem (TSP): given an edge-weighted graph, find a shortest tour visiting all its vertices. Suppose that instead of computing a single optimal tour, we must produce $k$ tours subject to a diversity requirement, while minimizing their \textit{bottleneck cost}---the length of the longest tour. If too few short tours satisfy this requirement, longer tours must be included, and the longest of the $k$ tours may end up much more expensive than the optimal single tour. 

To measure the trade-off between solution quality and the desired number of solutions, we analyze the ratio between the length of a longest tour in the best set of $k$ diverse tours and the length of an optimal single tour. Taking the worst-case value of this ratio over all instances provides a measure of how much the demand for diversity impacts the cost of solutions. We call this measure the \emph{Price of Diversity (PoD)}, which we introduce more formally and more generally in Section~\ref{section.PoD}. 

\subsubsection{Our problems and results.} \label{sec:not+ter}
We study the price of diversity in the context of the TSP in a metric space (i.e., edge costs satisfy the triangle inequality). Our main focus is on the setting where diversity is enforced through edge-disjointness and $k$, the number of required solutions, equals 2. This is one of the simplest and most natural cases, yet deriving tight bounds is already nontrivial. More precisely, we consider the following problems. 

\begin{extthm}[\disjointSHP{}] 
    Given a weighted graph $G = (V, E)$ and two specified vertices $s, t \in V$, find a pair $(H_1, H_2)$ of edge-disjoint Hamiltonian $(s, t)$-paths that minimizes $cost(H_1, H_2):=\max(c(H_1), c(H_2))$. 
\end{extthm}

\begin{extthm}[\disjointTSP{}] 
    Given a weighted graph $G = (V, E)$, find a pair $(T_1, T_2)$ of edge-disjoint tours (that is, Hamiltonian cycles) that minimizes $cost(T_1, T_2):=\max(c(T_1), c(T_2))$. 
\end{extthm}

We establish the PoD of \disjointSHP{} and \disjointTSP{} for different metric spaces. Already in one dimension, these problems exhibit nontrivial behavior. In Section~\ref{sec.1}, we consider \disjointSHP{} on the line with uniformly spaced vertices and show that the PoD equals $\tfrac{13}{7}$ (approaching $\tfrac{8}{5}$ asymptotically). In Section~\ref{sec.2}, we extend the analysis to \disjointTSP{} on the circle with uniformly spaced vertices. 

In Section~\ref{sec.3}, we turn to arbitrary metrics and show that the price of diversity of \disjointSHP{} is exactly $3$, while the PoD of \disjointTSP{} is exactly $2$. 
Each of these bounds is matched by a constructive algorithm that produces a pair of edge-disjoint solutions whose bottleneck cost achieves the claimed PoD. Note that we are, in general, not interested in the running time of the algorithms; 
we briefly discuss efficiency considerations in Section \ref{sec.4}.  

\subsubsection{Related Work.} 
The generation of multiple solutions has been studied in many areas of computer science, including constraint programming \cite{hebrard2005finding,petit2015finding}, mixed integer optimization \cite{ahanor2024diversitree,danna2009select,glover2000scatter}, $k$-best enumeration \cite{eppstein2016k,hamacher1985k}, genetic algorithms \cite{do2022niching,gabor2018preparing}, social choice \cite{boehmer2021broadening}, and computational economics \cite{gupta2025balancing,Paris2016}. More recently, attention has shifted to producing not just many but \textit{diverse} solutions \cite{baste2019fpt,fomin2024diverse,hanaka2021finding,misra2024parameterized,shida2024finding}.

There are many ways to specify a diversity requirement for a collection of solutions. Of particular relevance in this work is the case when diversity is specified in terms of disjointness. For problems where the output is a set of edges, such as matching, spanning tree, or minimum $s$-$t$ cut, edge-disjointness serves as a natural diversity criterion. Recent works have addressed such settings \cite{de2023finding,fomin2024diverse,hanaka2021finding}.  

In the context of the TSP, several works have considered the problem of finding edge-disjoint tours in specific graph classes. We mention the works of Rowley and Bose~\cite{rowley1991edge} on De Bruijn graphs, L\"u and Wu~\cite{lu2019edge} on balanced hypercubes, and Alspach~\cite{alspach2008wonderful} on complete graphs. The peripatetic $k$-TSP~\cite{krarup1995peripatetic} similarly requires multiple edge-disjoint tours, aiming to minimize their total length. Another related problem is Two-TSP, where $2n$ points in the plane must be partitioned into two equal-sized sets, and the objective is to minimize the maximum TSP tour length over the two parts \cite{bereg2022two}. These works focus on computing multiple non-overlapping solutions but do not quantify the cost overhead incurred by the disjointness requirement.

In this work, we introduce the price of diversity (PoD) as a measure of the tradeoff between solution quality and the requirement to compute multiple diverse solutions. The term has appeared before, especially in computational social choice, but usually refers to the cost of enforcing diversity within a \emph{single} solution---for example, in matching, assignment, clustering, or committee selection \cite{ijcai2017p6,benabbou2020price,thejaswi2021diversity,bredereck2018multiwinner}. In contrast, we analyze the PoD in the setting where a \textit{collection} of diverse solutions is required. To our knowledge, this is the first formal treatment of this perspective.

\section{Preliminaries} \label{section.PoD}

\subsubsection{The PoD Framework.} Let $\Pi$ be a minimization problem---like TSP---where each instance $I$ has an associated set of \textit{feasible solutions} $\Gamma(I)$. Each solution $s \in \Gamma(I)$ has a cost $c(s)$, and the objective is to find an optimal solution; that is, a feasible solution with minimum cost, $\mbox{OPT}(I)$. Consider now the following \textit{diverse} variant of $\Pi$, called simply $\Pi_k$, which, given an instance $I$ of problem $\Pi$ and a fixed integer $k > 0$, is tasked with finding a \textit{minimum-cost} set of $k$ feasible solutions to $I$ that satisfies a specified \textit{diversity requirement}. 

\begin{extthm}[\textsc{Problem} $\Pi_k$]
    Given an instance $I$ of $\Pi$, find a set $S$ of $k$ feasible solutions that (i) satisfies a \textit{diversity requirement} and (ii) has minimum \textit{cost}. 
\end{extthm}
        
Depending on the problem, there are many ways to define a diversity requirement. For example, for problems where solutions are subsets of a larger set, one could measure the diversity of a set of solutions as the cardinality of their union, the average size of their pairwise symmetric differences, or the minimum size among these differences. A requirement can then be imposed for the diversity to be at least a given value $t$; compare this with the edge-disjointness requirement for \textsc{TSP} in our earlier example. Similarly, the cost of a collection of feasible solutions can be defined in different ways. For instance, in \textsc{TSP}, one might measure cost as the length of the longest tour in the collection, or as the average length of all tours.

Returning to problem $\Pi_k$, let $\mathcal{S}(I) \in \Gamma(I)^k$ represent the collection of all $k$-sized sets of feasible solutions for instance $I$ that satisfy the diversity requirement, and let $cost(S)$ denote the cost of a set $S \in \mathcal{S}(I)$ of feasible solutions. Further, let $\mathcal{I}(\Pi)$ denote the set of all instances of problem $\Pi$. 
We define the Price of Diversity of problem $\Pi_k$, abbreviated with $PoD(\Pi_k)$, as 
\begin{equation} \label{eq.PoD}
    PoD(\Pi_k) = \underset{I \in \mathcal{I}(\Pi)}{\mbox{sup}} \left\{ \frac{ \min_{S \in \mathcal{S}(I)} cost(S) }{OPT(I)}\right\}.
\end{equation}

In plain words, $PoD(\Pi_k)$ represents the supremum, taken over all instances of problem $\Pi$, on the ratio between the minimum cost among diverse sets of feasible solutions and the cost of a single optimal solution. Note that we have assumed that $\Pi$ is a minimization problem, but this is without loss of generality as a similar measure to \eqref{eq.PoD} can be defined for maximization problems in a straightforward manner.

In some cases, as we shall see in Sections \ref{sec.1} and \ref{sec.2}, analyzing the \textit{asymptotic} worst-case behavior of $PoD(\Pi_k)$ provides a more meaningful perspective on the trade-off between diversity and solution quality. This is because the supremum of the ratio in \eqref{eq.PoD} may be dictated by a small instance, making $PoD(\Pi_k)$ less representative of its behavior across all instances. To address this, we define the \textit{asymptotic price of diversity} $PoD^\infty(\Pi_k)$ as

\begin{equation} \label{eq.PoDinf}
    PoD^\infty(\Pi_k)=\lim\limits_{n \to \infty}\sup\limits_{I\in \mathcal{I}(\Pi), |I|=n} \left\{  \frac{\min_{S \in\mathcal{S}(I)}cost(S)}{OPT(I)}\right\}.
\end{equation}

While $PoD$ and $PoD^\infty$ are properties corresponding to a particular problem (including a specification of diversity), we also want to be able to express the quality of an algorithm that outputs $k$ diverse solutions in terms of the quality of the set of diverse solutions as used above. Indeed, let $A$ be an algorithm for problem $\Pi_k$, and let $A(I)$ denote the set of $k$ solutions that algorithm $A$ produces when given instance $I$ of problem $P$. We define both the \textit{loss ratio} $LR_A$, and the \textit{asymptotic loss ratio} $LR^\infty_A$, of algorithm $A$ as 
\[
    LR_A = \underset{I \in \mathcal{I}(\Pi)}{\mbox{sup}} \left\{ \frac{cost(A(I))}{OPT(I)} \right\} \text{ and } LR^\infty_A = \lim_{n \rightarrow \infty} \sup\limits_{I\in \mathcal{I}(\Pi), |I|=n} \left\{ \frac{cost(A(I))}{OPT(I)} \right\}.
\]

\subsubsection{Notation.}
For a weighted graph $G = (V, E)$, let $n := |V|$ and $m := |E|$ denote the number of vertices and edges, respectively. If $G$ is a complete graph whose nodes correspond to points in a metric space and whose edge weights correspond to the distances between these points, then we call $G$ a \emph{metric graph}. We use $\mathcal{M}$ to denote the class of all metric graphs. Note that this is simply the class of complete weighted graphs whose edge weights satisfy the triangle inequality (and are non-negative). We are especially interested in $\mathcal{R}$, the class of metric graphs in $\mathbb{R}^1$, and in $\mathcal{S}$, the class of metric graphs in $\mathbb{S}^1$, where distances are measured using standard Euclidean geometry. (Here $\mathbb{S}^1$ denotes the circular 1-dimensional space.) We define $R_n^1\in \mathcal{R}$, the \emph{uniform metric graph on $n$ vertices in $\mathbb{R}^1$}, to be the metric graph whose vertices correspond to the set $[n] \subset \mathbb{R}^1$. Similarly, we define $\uniformCircleGeomGraph{} \in \mathcal{S}$, the \emph{uniform metric graph on $n$ vertices in $\mathbb{S}^1$}, to be the metric graph whose vertices correspond to evenly spaced points in a circle of length $n$, such that the angle between consecutive vertices is $2\pi/n$ and the distance (arc length) between them is 1. 
For a metric graph in $\mathcal{R}$ (or $\mathcal{S}$), we refer to the edges that connect consecutive vertices along $\mathbb{R}^1$ (or $\mathbb{S}^1$) as \textit{segments}.

\section{The Price of Diversity of \disjointSHP{} in \uniformLine{}} \label{sec.1}

In this section, we establish both the regular and the asymptotic price of diversity of \disjointSHP{} for the class of uniform metric graphs in $\mathbb{R}^1$, specifically when $s = 1$ and $t = n$, thereby proving the following. 

\begin{restatable}[]{theorem}{mainTheoremTwo} \label{thm:2}
    $PoD(\disjointSHP{}) = \frac{13}{7}$, and $PoD^\infty(\textsc{Disjoint-Path}$\\ $\textsc{TSP}_2) = \frac{8}{5}$ for the class of uniform metric graphs in $\mathbb{R}^1$ when $s = 1$ and $t = n$. 
\end{restatable}

Throughout this section, we take the uniform metric graph $R_n^1 = (V, E)$ as the input graph and define $P_n$ as the path graph formed by the vertex set $V$ and the segments of $R_n^1$. We say that $P_n$ is the path graph induced by $R_n^1$. We assume that the vertices of $P_n$ are labeled from left to right $(v_1, v_2, \ldots, v_n)$, with $s = v_1$ and $t = v_n$. Since we assume that all Hamiltonian $(s, t)$-paths in this section have the same endpoints, we omit explicit reference to $s$ and $t$ when referring to a Hamiltonian path. For a subset $X \subseteq V$ of vertices, we use $R_n^1[X]$ to denote the subgraph of $R_n^1$ induced by the vertex set $X$. For a Hamiltonian path $H$ in $R_n^1$, we use $H[X]$ (with a slight abuse of notation) to denote the (possibly disconnected) subgraph induced by $X$ of the path graph defined by $H$.

We prove Theorem \ref{thm:2} as follows. In Section~\ref{subsec.2.2}, we first establish a lower bound of $\frac{8}{5}$ for $PoD(\disjointSHP{})$. Then, in Section~\ref{subsec.2.1}, we develop an algorithm whose asymptotic loss ratio matches this bound, thus showing that $PoD^\infty(\disjointSHP{}) = \frac{8}{5}$. Also in Section~\ref{subsec.2.1}, we establish a loss ratio of $\frac{13}{7}$ for the algorithm and argue that this coincides with the price of diversity of \disjointSHP{}, thereby completing the proof. 

\subsection{A Lower Bound} \label{subsec.2.2}
We prove that $\frac{8}{5}$ is a lower bound to $PoD(\disjointSHP{})$ for the class of uniform metric graphs in $\mathbb{R}^1$. First, we introduce some new terminology. Let $v_i$ be a vertex in a path $H$ on $\uniformLineGeomGraph$. We say that $v_i$ is \textit{covered} by $H$ if there is an edge $(v_j, v_k) \in H$ such that $j < i < k$. Conversely, the edge $(v_j, v_k)$ is said to \textit{cover} the vertex $v_i$. The concept extends naturally to segments. An edge $(v_j, v_k) \in H$ is said to cover the segment $(v_i, v_{i+1}) \in P_n$ if $j \leq i$ and $i+1 \leq k$. The \textit{depth} $\mu_H(s)$ of a segment $s$ in $P_n$ w.r.t. path $H$ is the number of edges from $H$ that cover it. Note that the cost of $H$ is equal to the sum of the depths of each segment w.r.t. $H$. 
%


%
We now introduce the notions of \textit{$\ell$-piece} and \textit{cut-point}, which play a key role in the proof of the lower bound. Any sequence of $\ell$ consecutive segments in $P_n$ is called an $\ell$-piece. The depth of a piece (w.r.t. a path $H$) is the sum of the depths of its constituting segments. A pair of edge-disjoint Hamiltonian paths $(H_1, H_2)$ is said to have a cut-point $v$ if $v$ is not covered by $H_1$ nor by $H_2$.\footnote{Not to be confused with a \textit{cut-vertex}, which, in standard graph theory terminology, refers to a vertex whose removal results in a disconnected graph.} The \textit{total depth} of an $\ell$-piece w.r.t. the pair $(H_1, H_2)$ is the sum of depths w.r.t. $H_1$ and $H_2$, respectively. 
%



%

With this terminology in place, we can establish the following connection between cut-points and the total cost of a pair of Hamiltonian paths. Due to space limitations, the proof is deferred to Appendix~\ref{appendix.sec2.Lemma.1}.

\begin{restatable}[]{lemma}{lemmaNoCutpoint} \label{lemma.2}
    Consider a pair of disjoint Hamiltonian paths on the uniform metric graph \uniformLineGeomGraph. If they have no cut-point, then the sum of their costs is at least $\frac{16}{5} (n-1)$ for $n \geq 6$. 
\end{restatable}

In addition, we have verified by computer the following claim for small uniform metric graphs in $\realLine$.\footnote{Code available at \hyperlink{https://github.com/andoresu47/Disjoint-Tours-and-the-PoD}{https://github.com/andoresu47/Disjoint-Tours-and-the-PoD}.} Part (i) of the observation is also stated by Ageev \textit{et al.}~\cite{ageev20072}.

\begin{observation} \label{claim:1}
Consider the uniform metric graph $\uniformLineGeomGraph$ for some $n \leq 8$. Then (i) there is no pair of edge-disjoint Hamiltonian paths when $n \leq 5$, and (ii) there is no pair of edge-disjoint Hamiltonian paths with total cost less than $\frac{16(n - 1)}{5}$ when $n \in \{6, 7, 8\}$. 
\end{observation}

With Lemma~\ref{lemma.2} and Observation \ref{claim:1}, we are ready to prove the section’s main result. 

\begin{lemma} \label{theorem.2}
    There is no pair of disjoint Hamiltonian paths in the uniform metric graph $\uniformLineGeomGraph$ such that the sum of their costs is less than $\frac{16}{5} (n-1)$ for $n \geq 6$. 
\end{lemma}
\begin{proof}
By induction on $n$. For the base case, $n = 6$, Observation \ref{claim:1} shows that the statement is true. Now, assume that the statement holds for any integer $k$ such that $6 < k < n$, with $n > 6$. We will prove that, under this assumption, the statement is also true for the input graph \uniformLineGeomGraph{} of size $n$. For the sake of contradiction, suppose that there is a pair $(H_1, H_2)$ of edge-disjoint Hamiltonian paths of total cost less than $16(n-1)/5$. By Lemma \ref{lemma.2}, the pair $(H_1, H_2)$ has a cutpoint $v$. We show that this gives a contradiction. 

Let $P_n$ denote the path graph induced by \uniformLineGeomGraph. 
If $H_1$ and $H_2$ have a cutpoint $v$, we can partition $P_n$ into the two path graphs $P_{n_1} = (V_1, E_1)$ and $P_{n_2} = (V_2, E_2)$---on $n_1$ and $n_2$ vertices, respectively---such that $s \in P_{n_1}$, $t \in P_{n_2}$, and $V_1 \cap V_2 = \{v\}$. Note that $n = n_1 + n_2 - 1$. We define the total cost of a pair of paths $(H_1, H_2)$ as $C_\mathrm{tot}(H_1, H_2):=c(H_1) + c(H_2)$, i.e., the sum of the costs of the paths. 
There are two cases to consider: Either $\min(n_1, n_2) < 6$ or $\min(n_1, n_2) \geq 6$. 
\begin{itemize}
        \item[$\bullet$] Let $\min(n_1, n_2) < 6$. W.l.o.g., assume that $n_1 < 6$. In other words, we assume the shorter subpath is $H_1[V_1]$. By part (i) of Observation \ref{claim:1}, we know that no pair of disjoint Hamiltonian paths exists. This means the subpath $H_2[V_1]$ cannot be edge-disjoint with $H_1[V_1]$. But we are given that $H_1 = H_1[V_1] \cup H_1[V_2]$ and $H_2 = H_2[V_1] \cup H_2[V_2]$ are edge-disjoint. Hence, we reach a contradiction. 
        
        \item[$\bullet$] Let $\min(n_1, n_2) \geq 6$. The inductive hypothesis implies that $C_\mathrm{tot}(H_1[V_1], H_2[V_1]) \geq 16(n_1 - 1)/5$ and $C_\mathrm{tot}(H_1[V_2], H_2[V_2]) \geq 16(n_2 - 1)/5$. Then $C_\mathrm{tot}(H_1, H_2) = C_\mathrm{tot}(H_1[V_1], H_2[V_1]) + C_\mathrm{tot}(H_1[V_2], H_2[V_2]) \geq 16(n_1 + n_2 - 2)/5 = 16(n - 1)/5$. But this contradicts our initial assumption that the total cost of $H_1$ and $H_2$ was strictly less than $16(n - 1)/5$.
    \end{itemize}
Both cases result in a contradiction; therefore, the statement is true for $n > 6$ and the theorem is proven. \hfill$\qed$
\end{proof}

Notice that Lemma \ref{theorem.2} implies that in any pair of Hamiltonian paths on the uniform metric graph \uniformLineGeomGraph, the longest of the two paths has a cost of at least $\frac{8}{5} (n-1)$. This can be restated as the following corollary. 

\begin{corollary} \label{cor:lowerBoundPoDPaths}
     $PoD(\disjointSHP{}) \geq \frac{8}{5}$ for uniform metric graphs in $\mathbb{R}^1$ when $s = 1$ and $t = n$. 
\end{corollary}

\subsection{An Algorithm for \disjointSHP{}} \label{subsec.2.1}


We now present an algorithm for \disjointSHP{} on the uniform metric graph \uniformLineGeomGraph, which achieves a loss ratio of $\tfrac{13}{7}$ and an asymptotic loss ratio of $\tfrac{8}{5}$. The main idea is to \emph{concatenate} small optimal building blocks of disjoint paths. The algorithm is detailed below as Algorithm \ref{algo.approxPaths}. A more formal description is provided in Appendix \ref{appendix.algorithmProofs}. 

When the number of segments, $n-1$, is a multiple of 10, a pair of edge-disjoint Hamiltonian paths of optimal bottleneck cost $\tfrac{8}{5}(n-1)$ can be constructed by concatenating $k$ copies of an optimal solution for $n=11$. One such pair is shown in Figure~\ref{fig:OptPathsN11}. For general $n$, we extend this idea by attaching suitable smaller optimal solutions for $6\leq n\leq 10$ (found by computer-aided exhaustive search) or by slightly augmenting the $n=11$ construction. The illustrations of optimal solutions for $n \in \{6, 7, 8, 9, 10\}$, as well as the augmented constructions for $n \in \{12, 13, 14, 15\}$, are deferred to Appendices \ref{appendix.optPaths} and \ref{appendix.optPaths2}, respectively. 

\begin{algorithm}[!ht]
\caption[Caption for LOF]{Algorithm \texttt{Paths}} \label{algo.approxPaths}
\vspace{.5em}
\textbf{Input:} Uniform metric graph \uniformLineGeomGraph \\ 
\textbf{Output:} A pair of disjoint Hamiltonian paths. 
\algrenewcommand\algorithmiccomment[1]{\hfill {\color{blue} \(\triangleright\) #1}}
\vspace{.5em}
\begin{algorithmic}[1]
\If{$n \leq 11$} 
\State \Return an optimal pair (precomputed).
\Else 
\State Let $(n-1)=10k+\ell$. 
\If{$\ell=0$} 
\State \Return $k$ concatenated copies of the $n=11$ solution of Figure~\ref{fig:OptPathsN11}.
\ElsIf{$5\leq \ell \leq 9$} 
\State Concatenate $k$ copies of the $n=11$ solution with a precomputed\\
\hspace{28pt}(optimal) solution for $n = \ell+1$ (see Appendix \ref{appendix.optPaths}).
\Else \Comment{$\ell \in \{1, 2, 3, 4\}$}
\State concatenate $k - 1$ copies of the $n=11$ solution with a precomputed solution \\
\hspace{28pt}for $n = 11 + \ell$ (see Appendix \ref{appendix.optPaths2}).
\EndIf
\EndIf 
\end{algorithmic}
\vspace{.5em}
\end{algorithm}

\begin{figure}[t]
    \centering
    \begin{tikzpicture}[scale=0.8, transform shape]
\node [draw=black,circle,inner sep=0pt,minimum size=1.25em] (v3) at (0,1.25) {\footnotesize $v_{3}$};
\node [draw=black,circle,inner sep=0pt,minimum size=1.25em] (v4) at (0.75,1.25) {\footnotesize $v_{4}$};
\node [draw=black,circle,inner sep=0pt,minimum size=1.25em] (v5) at (1.5,1.25) {\footnotesize $v_{5}$};
\node [draw=black,circle,inner sep=0pt,minimum size=1.25em] (v6) at (2.25,1.25) {\footnotesize $v_{6}$};
\node [draw=black,circle,inner sep=0pt,minimum size=1.25em] (v2) at (-0.75,1.25) {\footnotesize $v_{2}$};
\node [draw=black,circle,inner sep=0pt,minimum size=1.25em] (v1) at (-1.5,1.25) {\footnotesize $v_{1}$};
\node [draw=black,circle,inner sep=0pt,minimum size=1.25em] (v7) at (3,1.25) {\footnotesize $v_{7}$};
\node [draw=black,circle,inner sep=0pt,minimum size=1.25em] (v8) at (3.75,1.25) {\footnotesize $v_{8}$};
\node [draw=black,circle,inner sep=0pt,minimum size=1.25em] (v9) at (4.5,1.25) {\footnotesize $v_{9}$};
\node [draw=black,circle,inner sep=0pt,minimum size=1em] (v10) at (5.25,1.25) {\footnotesize $v_{10}$};
\node [draw=black,circle,inner sep=0pt,minimum size=1.25em] (v11) at (6,1.25) {\footnotesize $v_{11}$};
\draw  (v1) edge[bend left=45] (v3);
\draw  (v3) edge (v4);
\draw  (v4) edge[bend right=45] (v2);
\draw  (v2) edge[bend right=45] (v5);
\draw  (v5) edge (v6);
\draw  (v6) edge (v7);
\draw  (v7) edge[bend left=45] (v9);
\draw  (v9) edge (v8);
\draw  (v8) edge[bend left=45] (v10);
\draw  (v10) edge (v11);
\node at (-2.5,1.25) {$H_1: $};

\node [draw=black,circle,inner sep=0pt,minimum size=1.25em] (v3) at (0,-0.25) {\footnotesize $v_{3}$};
\node [draw=black,circle,inner sep=0pt,minimum size=1.25em] (v4) at (0.75,-0.25) {\footnotesize $v_{4}$};
\node [draw=black,circle,inner sep=0pt,minimum size=1.25em] (v5) at (1.5,-0.25) {\footnotesize $v_{5}$};
\node [draw=black,circle,inner sep=0pt,minimum size=1.25em] (v6) at (2.25,-0.25) {\footnotesize $v_{6}$};
\node [draw=black,circle,inner sep=0pt,minimum size=1.25em] (v2) at (-0.75,-0.25) {\footnotesize $v_{2}$};
\node [draw=black,circle,inner sep=0pt,minimum size=1.25em] (v1) at (-1.5,-0.25) {\footnotesize $v_{1}$};
\node [draw=black,circle,inner sep=0pt,minimum size=1.25em] (v7) at (3,-0.25) {\footnotesize $v_{7}$};
\node [draw=black,circle,inner sep=0pt,minimum size=1.25em] (v8) at (3.75,-0.25) {\footnotesize $v_{8}$};
\node [draw=black,circle,inner sep=0pt,minimum size=1.25em] (v9) at (4.5,-0.25) {\footnotesize $v_{9}$};
\node [draw=black,circle,inner sep=0pt,minimum size=1.25em] (v10) at (5.25,-0.25) {\footnotesize $v_{10}$};
\node [draw=black,circle,inner sep=0pt,minimum size=1.25em] (v11) at (6,-0.25) {\footnotesize $v_{11}$};
\draw  (v1) edge (v2);
\draw  (v2) edge (v3);
\draw  (v3) edge[bend left=45] (v5);
\draw  (v5) edge (v4);
\draw  (v4) edge[bend left=45] (v6);
\draw  (v6) edge[bend left=45] (v8);
\draw  (v8) edge (v7);
\draw  (v7) edge[bend right=45] (v10);
\draw  (v10) edge (v9);
\draw  (v9) edge[bend left=45] (v11);
\node at (-2.5,-0.25) {$H_2:$};

\end{tikzpicture}
    \caption{An optimal pair of edge-disjoint Hamiltonian paths with $cost(H_1, H_2) = 16$ 
    for $n = 11$.}
    \label{fig:OptPathsN11}
\end{figure}

\paragraph{Analysis.} 
Note that each time the algorithm concatenates pairs of disjoint paths, a cut-point is created, namely, the vertex joining the concatenated paths. Consider the pair of paths returned by the algorithm. By construction, the pair of subpaths induced by the vertices between cutpoints are edge-disjoint. Hence, the paths in the returned pair are also edge-disjoint. 

\begin{lemma} \label{theorem.algo1}
    Algorithm \ref{algo.approxPaths} generates a pair of edge-disjoint Hamiltonian paths. 
\end{lemma}

We show that its loss ratio is $\tfrac{13}{7}$ and its asymptotic loss ratio is $\tfrac{8}{5}$. Due to space constraints, we defer the proofs to Appendix~\ref{appendix.algorithmProofs}, but we sketch the key ideas here. For small instances, the worst case occurs at $n=8$, where the ratio $\tfrac{13}{7}$ is attained. For larger $n$, the ratio decreases, converging to $\tfrac{8}{5}$ as $n\to\infty$. 

\begin{lemma} \label{theorem:LRPathsCombined}
    Algorithm \ref{algo.approxPaths} has loss ratios $LR_{\texttt{Paths}} = \frac{13}{7}$ and $LR^\infty_{\texttt{Paths}} = \frac{8}{5}$.
\end{lemma}

The loss ratio of Algorithm \ref{algo.approxPaths} provides an upper bound for $PoD(\textsc{Disjoint-}$\\$\textsc{Path TSP}_2)$ with $s=1$ and $t=n$. The proof of Lemma \ref{theorem:LRPathsCombined} actually shows this bound is tight, because the ratio $\frac{\min_{S \in \mathcal{S}(\mathcal{I})} \{\text{cost}(S) \}}{OPT(I)}$ achieved by the optimal solution at $n=8$ is already an upper bound on the loss ratio for all other values of $n$.

\section{The Price of Diversity of \disjointTSP{} in \uniformCircle{}} \label{sec.2}

We now extend the results and concepts of the previous section to \disjointTSP{} in \uniformCircle{}. We establish the following.

\begin{restatable}[]{theorem}{mainTheorem} \label{thm:1}
    $PoD(\disjointTSP{}) = \frac{13}{7}$, and $PoD^\infty(\disjointTSP{}) = \frac{8}{5}$ for the class of uniform metric graphs in $\mathbb{S}^1$.
\end{restatable}

We prove Theorem \ref{thm:1} as follows. In Section \ref{subsec:lowerBoundUnitCycles}, we extend the lower bound established in Corollary \ref{cor:lowerBoundPoDPaths} to pairs of edge-disjoint tours. Subsequently, in Section \ref{subsec:TwoAlgos}, we present an algorithm that builds on the results of Section \ref{sec.1}, generalizing the approach from paths to tours while maintaining the loss ratio. Together, these results provide the proof of Theorem \ref{thm:1}. Recall that a segment is an edge of $G$ that connects consecutive vertices along $\mathbb{S}^1$. Throughout this section, we use $C_n$ to denote the cycle graph formed by the vertex set $V$ and the segments of $G$. We refer to $C_n$ as the cycle graph induced by $G$. 

\subsection{A Lower Bound} \label{subsec:lowerBoundUnitCycles}
We now extend the lower bound results of Section \ref{subsec.2.2}---which focused on Hamiltonian ($s$, $t$)-paths, with $s = 1$ and $t = n$---to tours in uniform metric graphs. To this end, we adopt the notions of \textit{covered} vertices and segments introduced in Section~\ref{subsec.2.2}, modifying them by replacing the Hamiltonian path with a tour in each definition. As before, we denote the depth of a segment $s \in C_n$ w.r.t. a tour $T$ by $\mu_T(s)$. Note that the cost of $T$ is equal to the sum of the depths of all segments in $C_n$ w.r.t. $T$. We begin with the following crucial lemma (see proof in Appendix \ref{appendix.sec3.Lemma.6}). 

 \begin{lemma} \label{lemma:3} 
     In any tour of a metric graph in $\mathbb{S}^1$, all segments have odd depth or all segments have even depth. 
 \end{lemma}

 Consider the case when the input graph $G$ is the uniform metric graph \uniformCircleGeomGraph{}. Lemma \ref{lemma:3} implies that any tour that covers some segment of \uniformCircleGeomGraph{} an even number of times has a total cost that is roughly twice the optimal value. (Note that we treat 0 as an even number.)

 \begin{corollary} \label{corollary:1}
    Any pair of disjoint tours $T_1$ and $T_2$ in $\uniformCircleGeomGraph$ such that one of them has exclusively even depth segments satisfies $\max(c(T_1), c(T_2)) \geq 2 \cdot (n - 1)$. 
 \end{corollary}

 Hence, only pairs of tours in \uniformCircleGeomGraph{} with exclusively odd depth segments may have a bottleneck cost smaller than $2 \cdot (n - 1)$. We call such tours \textit{odd-depth} tours. Similarly, \textit{even-depth} tours are tours with exclusively even-depth segments. In the rest of the section, we restrict our discussion to pairs of \textit{odd-depth} tours. 
 


We now extend the notions of \textit{cut-point} and \textit{$\ell$-piece} from Hamiltonian paths to tours. A pair of edge-disjoint tours $(T_1, T_2)$ is said to have a \textit{cut-point} $v$ if $v$ is not covered by $T_1$ nor by $T_2$. Any sequence of $\ell$ consecutive segments in $C_n$ is called an \textit{$\ell$-piece}. The depth of a piece (w.r.t. a tour $T$) is the sum of the depths of its constituting segments. The \textit{total depth} of an $\ell$-piece w.r.t. the pair $(T_1, T_2)$ is the sum of depths w.r.t. $T_1$ and $T_2$, respectively. 

With these definitions, we establish the following analogue of Lemma \ref{lemma.2} for odd-depth tours (see proof in Appendix \ref{appendix.sec3.Lemma.7}). 




\begin{lemma} \label{lemma:4}
    Consider a pair of odd-depth disjoint tours on the uniform metric graph $\uniformCircleGeomGraph$. If they have no cut-point, then the sum of their costs is at least $\frac{16}{5} n$ for $n \geq 5$. 
\end{lemma}

With Lemma \ref{lemma:4}, we have the necessary ingredients to extend the result of Lemma \ref{theorem.2} from pairs of Hamiltonian paths to pairs of odd-depth tours. 

\begin{lemma} \label{theorem.3}
    There is no pair of disjoint odd-depth tours in the uniform metric graph $\uniformCircleGeomGraph$ such that the sum of their costs is less than $\frac{16}{5} n$ for $n \geq 5$. 
\end{lemma}
\begin{proof}
Consider an arbitrary pair of odd-depth tours in $\uniformCircleGeomGraph$. By contrapositive of Lemma \ref{lemma:4}, if the sum of their costs is less than $16n/5$, then either they are not disjoint, or they have a cut-point, or both. We are interested in low-cost pairs of disjoint tours, so we may only consider the case of the tours having a cutpoint. If such a pair $(T_1, T_2)$ of disjoint tours existed, then we could construct two disjoint Hamiltonian $(s, t)$-paths $H_1$ and $H_2$ of total cost less than $16n/5$ as follows. Let $v$ denote the cutpoint of $T_1$ and $T_2$. We construct $H_1$ by creating two copies, $s$ and $t$, of vertex $v$ and replacing the edges $(u_1, v), (v, w_1) \in T_1$ with $(u_1, t)$ and $(s, w_1)$, respectively. We proceed similarly for path $H_2$ and edges $(u_2, v), (v, w_2) \in T_2$. The ($s$, $t$)-paths $H_1$ and $H_2$ just created are edge-disjoint since the edges we replaced from $T_1$ and $T_2$ were already disjoint. Further, the new edges have the same length as the ones replaced. Hence, the total cost of the paths we created is the same as the total cost of the original tours: less than $16n/5$. This is a pair of Hamiltonian paths on the uniform metric graph $\uniformLineGeomGraph$ whose total cost is $16n/5$. 
But, by Theorem \ref{theorem.2}, we know that no such pair of Hamiltonian paths exists. We arrive at a contradiction. Hence, the theorem is proved.   \hfill$\qed$
\end{proof}

Putting Corollary \ref{corollary:1} and Lemma \ref{theorem.3} together, we get the following result for pairs of disjoint tours, regardless of their depth. 

\begin{corollary} \label{lemma:5}
    There is no pair of disjoint tours in the uniform metric graph $\uniformCircleGeomGraph$ such that the cost of the longer tour is less than $\frac{8}{5}n$ for $n \geq 5$. 
\end{corollary}
\begin{proof}
    Let $(T_1, T_2)$ be an arbitrary pair of disjoint tours in $\uniformCircleGeomGraph$. If at least one of $T_1$ or $T_2$ is an even-depth tour, then by Corollary \ref{corollary:1}, we have $cost(T_1, T_2) \geq 2 (n-1)\geq \frac{8}{5}n$ for $n \geq 5$. If both $T_1$ and $T_2$ are odd-depth tours, Lemma \ref{theorem.3} implies $cost(T_1, T_2) \geq \frac{16}{10}n = \frac{8}{5}n$ for for $n \geq 5$ as well. \hfill$\qed$ 
\end{proof}

Finally, we get the following as an immediate consequence of Corollary \ref{lemma:5}. 

\begin{corollary}
$PoD(\disjointTSP{}) \geq \frac{8}{5}$ for the class of uniform metric graphs in $\realCircle$.
\end{corollary}

\subsection{An Algorithm for \disjointTSP{}} \label{subsec:TwoAlgos}

We now describe an algorithm, referred to as \texttt{Tours}, that generates two disjoint tours such that the cost of the longer tour is always within a factor $\frac{13}{7}$ of the cost of an optimal tour, and asymptotically within a factor $\frac{8}{5}$. A more detailed description of the algorithm is deferred to Appendix \ref{appendix.algoToursDescription}. 

The algorithm is an extension of Algorithm \ref{algo.approxPaths} for \disjointSHP{}. It works by first using the earlier algorithm as a subroutine to generate two edge-disjoint Hamiltonian paths, each of length $n$ $\Rightarrow$ $n+1$. Each path is then transformed into a tour by contracting its first and last vertices into a single new vertex. As a result, the cycles each contain exactly $n$ vertices and $n$ edges. Note that the newly created vertex is a cutpoint, ensuring that the cycles remain edge-disjoint, given that the original paths were edge-disjoint. Furthermore, the loss ratio of the algorithm inherits the bounds established for Algorithm \ref{algo.approxPaths}.

\begin{theorem}
    Algorithm \texttt{Tours} has loss ratio $LR_\texttt{Tours} = \frac{13}{7}$ and asymptotic loss ratio $LR^\infty_\texttt{Tours} = \frac{8}{5}$. 
\end{theorem}
\begin{proof}
    Replace $n$ with $n+1$ in the proof of Lemma \ref{theorem:LRPathsCombined}. 
\end{proof}

\section{The PoD of \disjointSHP{} and \disjointTSP{} in \generalClassGeomGraphs{}} \label{sec.3}

We conclude our investigation of the PoD by examining \disjointSHP{} and \disjointTSP{} in general metric graphs. First, in Section \ref{subsec:lowerBoundsGenMetric}, we derive lower bounds of $3 - \varepsilon$ and $2 - \varepsilon$, for any $\epsilon > 0$, on the price of diversity for these problems, respectively. Next, in Section \ref{subsec:algosGenMetric}, we present two simple algorithms whose loss ratios essentially match these lower bounds, thus establishing the price of diversity of the problems. 

\begin{restatable}[]{theorem}{mainTheoremFour} \label{thm:4}
    $PoD(\disjointSHP{}) = 3$ and $PoD(\disjointTSP{}) = 2$ in $\mathcal{M}$.
\end{restatable}

\subsection{Lower bounds} \label{subsec:lowerBoundsGenMetric}

In this section, we present lower bounds for the PoD of both \disjointSHP{} (Section \ref{subsubsection.1}) and \disjointTSP{} (Section \ref{subsubsection.2}) for general metric graphs. 

\subsubsection{\disjointSHP{}.} \label{subsubsection.1}
Consider a pair of edge-disjoint Hamiltonian ($s$, $t$)-paths $H_1$ and $H_2$ in a (not necessarily uniform) metric graph $G = (V, E)$ in $\mathbb{R}^1$, where $s$ and $t$ are the first and last vertices of $G$ along the line $\mathbb{R}^1$. Like in Section \ref{sec.1}, we use $P_n$ to denote the path graph induced by $G$ and assume that its vertices are labeled from left to right $(v_1, v_2, \ldots, v_n)$, with $s = v_1$ and $t = v_n$. We make the following claim (see proof in Appendix \ref{appendix.sec4.Claim.1}). 

\begin{claim} \label{claim:2}
    The depth of the segment $(v_2, v_3)$ (resp. $(v_{n-2}, v_{n-1})$) w.r.t. at least one of $H_1$ and $H_2$ is at least 3. 
\end{claim}

We now use this claim to derive a lower bound for the PoD of \disjointSHP{} in a general metric. We show that there is a metric graph on $n$ points in $\realLine$ where the longest path of an arbitrary pair of disjoint Hamiltonian paths has a cost of almost three times the optimal value.  

\begin{lemma} \label{theorem.4}
    For any $\varepsilon > 0$, the PoD of \disjointSHP{} in a general metric is at least $3 - \varepsilon$.
\end{lemma}
\begin{proof}
    Consider the metric graph $G = (V, E)$ in $\realLine$ of size $n > 5$, where $V = \{v_1, \ldots, v_n\}$ and edge weights are defined as follows: $w(v_i, v_{i+1}) = 1$ for all $i \neq 2$, and $w(v_i, v_{i+1}) = W$ for $i = 2$. 
    Now, let $H_1$ and $H_2$ be a pair of edge-disjoint Hamiltonian ($s$, $t$)-paths in $G$, with $s = v_1$ and $t = v_n$. By Claim \ref{claim:2}, the depth of segment $(v_2, v_3)$ is at least 3 in one of the two paths, which we assume w.l.o.g. to be $H_1$. Then, $c(H_1) \geq 3 \cdot W + (n - 2)$, while $OPT = W + (n - 2)$. Hence
    \begin{equation*}
        PoD(\disjointSHP{}) \geq \frac{3W+ (n - 2)}{W + (n - 2)} = 3 - \frac{2 n - 4}{W + n - 2}.
    \end{equation*}
    Thus, for any $\varepsilon > 0$, we can obtain a PoD of $3 - \varepsilon$ by setting $W:= \frac{(2- \varepsilon)(n-2)}{\varepsilon}$. \hfill$\qed$
\end{proof}

\begin{corollary}
    There is no algorithm for \disjointSHP{} in metric graphs with a loss ratio strictly less than 3. 
\end{corollary}

\subsubsection{\disjointTSP{}.} \label{subsubsection.2}

We now turn our attention to \disjointTSP{} in general metric graphs. We show that there is a metric graph in $\realCircle$ where the longest tour of an arbitrary pair of disjoint tours has a cost that is (almost) twice the optimal value. The proof is long and technical, so we provide only a sketch here; the full argument appears in Appendix~\ref{appendix.proofOfLemmaPoDTSP2}. 

\begin{restatable}[]{lemma}{poDofTSPTwo} \label{lemma:poDofTSP2}
    For any $\varepsilon > 0$, the PoD of \disjointTSP{} in a general metric is at least $2 - \varepsilon$. 
\end{restatable}
\begin{proofSketch}
The construction is similar in spirit to Lemma~\ref{theorem.4} for the path case. 
We consider a cycle graph on $n$ vertices where two specific edges are made very expensive, of weight $W$, while all other edges have unit weight. 
Any single optimal tour has cost roughly $2W$, but when two edge-disjoint tours are required, one of them is forced to traverse both expensive edges, incurring a cost of about $4W$. 
This yields a ratio approaching~2 as $W$ grows. 

The main technical challenge is to show that, regardless of whether the tours are odd- or even-depth, one of the two tours must indeed pick up both heavy edges. This is established through a series of structural claims about depth patterns of consecutive segments. \hfill$\qed$
\end{proofSketch}

\begin{corollary}
    There is no algorithm for \disjointTSP{} in $\mathcal{M}$ with a loss ratio strictly less than 2. 
\end{corollary}

\subsection{Algorithms} \label{subsec:algosGenMetric} 

We now present two simple algorithms, one for \disjointSHP{} and another for \disjointTSP{}, that achieve loss ratios of 3 and 2, respectively, in any weighted graph. 

\subsubsection{Algorithm for \disjointSHP{}.}  Let us first consider \textsc{Disjoint-Path} $\textsc{TSP}_2$, where the goal is to find a pair of edge-disjoint Hamiltonian $(s, t)$-paths in an (arbitrarily) weighted graph $G$. Let $H_\mathrm{opt} = (v_1, v_2, \ldots, v_n)$ be an optimal single Hamiltonian $(s, t)$-path in $G$. Consider now a metric graph $G'$ in $\realLine$ with the same vertex set as $G$, where the vertices are placed in the order given by $H_\mathrm{opt}$ and spaced so that the distance between $v_i$ and $v_{i+1}$ in $G'$ is equal to the weight of the segment $(v_i, v_{i+1})$ in $H_\mathrm{opt}$. Using Algorithm \ref{algo.approxPaths}, we can obtain a pair of edge-disjoint Hamiltonian $(s, t)$-paths in $G'$. 

What happens if we use the same pair of solutions but with the edge weight of the original graph $G$? We claim that the loss ratio of this algorithm is $3$. Indeed, this follows from the fact that the solutions generated by Algorithm \ref{algo.approxPaths} never cover any segment more than three times. Since the segments are defined by the optimal solution $H_\mathrm{opt}$, each of our edge-disjoint paths will have a cost at most three times the cost of $H_\mathrm{opt}$. 

\begin{lemma}
    The described algorithm has loss ratio 3. 
\end{lemma}

\subsubsection{An algorithm for \disjointTSP{}.} 
We now present a simple algorithm for \disjointTSP{} that achieves a loss ratio of 2. This algorithm is listed below as Algorithm \ref{algo.approxNaive}. 
We omit a detailed description as the provided pseudocode is self-explanatory. Figures \ref{fig:AlgoToursN7} and \ref{fig:AlgoToursN8} in Appendix \ref{appendix.toursAlgo} illustrate tours produced by the algorithm.
We defer the proof of Lemma \ref{lemma:6} below to Appendix \ref{appendix.sec4.Lemma.12}. 

\begin{algorithm}[!ht]
\caption[Caption for LOF]{Na\"ive $2$-loss ratio algorithm} \label{algo.approxNaive}
\vspace{.5em}
\textbf{Input:} Metric graph $G$ on $n > 4$ vertices. \\
\textbf{Output:} A pair $(T_1, T_2)$ of tours of minimum cost $\max(c(T_1), c(T_2))$. 
\vspace{.5em}
\begin{algorithmic}[1]
\newcommand\NoDo{\renewcommand\algorithmicdo{}}
\newcommand\ReDo{\renewcommand\algorithmicdo{\textbf{do}}}
\algrenewcommand\algorithmiccomment[1]{\hfill {\color{blue} \(\triangleright\) #1}}

\algdef{SE}[SUBALG]{Indent}{EndIndent}{}{\algorithmicend\ }%
\algtext*{Indent}
\algtext*{EndIndent}

\State Compute an optimal cost tour $H_\mathrm{opt} = (v_1, v_2, \ldots, v_{n}, v_1)$ in $G$.
\If{$n$ is odd} 
\State $T_1 \leftarrow H_\mathrm{opt}$
\State $T_2 \leftarrow (v_1, v_3, \ldots, v_{n-2}, v_n, v_2, v_4, \ldots, v_{n-1}, v_1)$
\EndIf
\If{$n$ is even} 
\State $T_1 \leftarrow (v_1, v_n, v_2, v_3, \ldots, v_{n-1}, v_1)$
\State $T_2 \leftarrow (v_1, v_3, \ldots, v_{n-1}, v_n, v_{n-2}, \ldots, v_2, v_1)$
\EndIf
\State Return the pair of tours $(T_1, T_2)$
\end{algorithmic}
\vspace{.5em}
\end{algorithm}

\begin{lemma} \label{lemma:6} 
    Algorithm \ref{algo.approxNaive} has a loss ratio of 2. 
\end{lemma}

The loss ratios achieved by the algorithms described in this section, together with the lower bounds of Section \ref{sec.4}, establish the PoD of \disjointSHP{} and \disjointTSP{} in general metric graphs. 

\section{Concluding remarks} \label{sec.4}
We introduced two problems \disjointSHP{} and \disjointTSP{}, and studied them on metric graphs. For the classes of uniform metric graphs in $\realLine$ and $\realCircle$, respectively, we presented $\frac{13}{7}$-loss ratio algorithms for both problems and proved a lower bound of $8/5$ for their PoD. Both of our algorithms are asymptotically tight, as they achieve an asymptotic loss ratio of $8/5$ when $n \rightarrow \infty$. For the class of general metric graphs, for any $\varepsilon > 0$, we show lower bounds of $3$ and $2$ for the PoD of \disjointSHP{} and \disjointTSP{}, respectively. We then provide two algorithms: Algorithm \ref{algo.approxPaths} for \disjointSHP{} and Algorithm \ref{algo.approxNaive} for \disjointTSP{}, whose loss ratios match these lower bounds, respectively. Hence, the PoD of \disjointSHP{} and \disjointTSP{} in a general metric are 3 and 2, respectively.

In the case of \disjointTSP{}, the naive algorithm relies on the computation of an optimal TSP solution, making it impractical in terms of running time. The algorithm can be made efficient (running in polynomial time); however, if one is willing to settle for a starting tour that is only approximate with respect to the optimal tour. For instance, starting the algorithm with a sub-optimal tour, namely one obtained via the Christofides $3/2$-approximation algorithm, results in a naive solution to \disjointTSP{} with approximation ratio 3. The question of whether one can design a polynomial time algorithm with loss-ratio $2 < \alpha < 3$ remains open. We have only presented results for $\textsc{Disjoint-Path TSP}_k$ and $\textsc{Disjoint-TSP}_k$ with $k = 2$. Hence, it is also natural to ask what the price of diversity is for these problems when considering $k > 2$. 

\subsubsection{Acknowledgements.} This research was supported by the European Union’s Horizon 2020 research and innovation programme under the Marie Skłodowska-Curie grant agreement no. 945045, and by the NWO Gravitation project NETWORKS under grant no. 024.002.003. We are grateful to the reviewers for their valuable comments and suggestions. 

\bibliographystyle{splncs04}
\bibliography{biblio}

\renewcommand{\theHsection}{A\arabic{section}}      
\appendix

\section{Proofs of Section \ref{sec.1}} \label{appendix.sec2}

\subsection{Proof of Lemma \ref{lemma.2}} \label{appendix.sec2.Lemma.1}
To prove the lemma, we begin with two simple observations. We only provide a proof of the first, as the second is immediate.

\begin{observation} \label{fact:2}
    The depth of a segment of $P_n$ w.r.t. any ($s$, $t$)-path $H$ is odd. 
\end{observation}
\begin{proof}
    Recall that $P_n = (v_1, v_2, \ldots, v_{n-1}, v_n)$, where $s = v_1$ and $t = v_n$. Consider an arbitrary segment $(v_i, v_{i+1})$ in $P_n$ with $i \in \{1, \ldots, n - 1\}$. Let $L = \{s,\ldots, v_i\}$ and $R = \{v_{i+1}, \ldots, t\}$. By definition of cover, every occurrence of an edge $e = (u, v)$ in $H$ with $u \in L$ and $v \in R$ means that $e$ covers segment $(v_i, v_{i+1})$. As $H$ starts in $s$ and ends in $t$, this must happen an odd number of times. \hfill$\qed$
\end{proof}

\begin{observation} \label{claim:0}
    Let $H$ be a Hamiltonian path. The segments $(s, v_2)$ and $(v_{n-1}, t)$ each have depth 1 w.r.t. $H$.
\end{observation}

We also require the following three claims. 

\begin{claim} \label{claim:6}
    Let $\gamma = (s_1, s_2, s_3)$ be a 3-piece that has total depth 3 w.r.t. some Hamiltonian path $H$. Then the middle segment $s_2$ is an edge of $H$.  
\end{claim}
\begin{proof}
 For the sake of contradiction, suppose that $s_2 = (v_i, v_{i+1})$ is not an edge of $H$. Hence, at least one of the vertices $v_i$ or $v_{i+1}$ must be covered by some edge of $H$. But every vertex in $H$ (except for its endpoints) has degree 2. This means that either segment $s_1$ or $s_3$ has a depth greater than 1 w.r.t. $H$. By Observation \ref{fact:2}, the depth of each segment of $\gamma$ w.r.t. $H$ must be odd. Therefore, $\gamma$ has a total depth greater than 3, which is a contradiction. \hfill$\qed$
\end{proof}

\begin{claim} \label{claim:7}
    Let $\gamma = (s_1, s_2, s_3)$ be a 3-piece that has total depth 5 w.r.t. some Hamiltonian path $H$. If the middle segment $s_2$ has depth 3 w.r.t. $H$, then $s_2$ is an edge of $H$.  
\end{claim}
\begin{proof}
    By Observation \ref{fact:2}, the depth of segments $s_1 = (v_i, v_{i+1})$ and $s_3 = (v_{i+2}, v_{i+3})$ must be 1 w.r.t. $H$. For the sake of contradiction, suppose that $s_2 = (v_{i+1}, v_{i+2})$ is not an edge of $H$. Then, without loss of generality, the vertex $v_{i+1}$ is covered by an edge of $v_{i+2}$; otherwise, the segment $s_3$ would have a depth greater than 1 w.r.t. $H$. Since every vertex in $H$ (except for its endpoints) has degree 2, and $v_{i+1}$ is not an endpoint of $H$, any edge of $v_{i+1}$ now covers at least one of the segments $s_1$ or $s_3$. But these were already covered by the edges of $v_{i+2}$. Therefore, at least one of $s_1$ or $s_3$ has a depth greater than 1 w.r.t. $H$, which is a contradiction. \hfill$\qed$
\end{proof}

\begin{figure}[H]
    \centering
    \begin{tikzpicture}[scale = 1.0, transform shape]
    \node[draw=black,circle,inner sep=0pt,minimum size=1.5em] (s) at (-3,1) {\footnotesize $s$};
    \node (v0) at (-2.25,1) {$\dots$};
\node[draw=black,circle,inner sep=0pt,minimum size=1.5em] (v3) at (0.5,1) {\footnotesize $v_{i+2}$};
\node[draw=black,circle,inner sep=0pt,minimum size=1.5em] (v4) at (1.5,1) {\footnotesize $v_{i+3}$};
\node (v5) at (2.25,1) {\footnotesize $\dots$};
\node[draw=black,circle,inner sep=0pt,minimum size=1.5em] (v2) at (-0.5,1) {\footnotesize $v_{i+1}$};
\node[draw=black,circle,inner sep=0pt,minimum size=1.5em] (v1) at (-1.5,1) {\footnotesize $v_{i}$};
\draw  (v1) edge[bend left=45] (v3);
\draw  (v3) edge (v2);
\draw  (v2) edge[bend right=45] (v4);
\node[draw=black,circle,inner sep=0pt,minimum size=1.5em] (s) at (3,1) {\footnotesize $t$};
\end{tikzpicture}
    \caption{A 3-piece with vertices $v_i, , v_{i+1}, v_{i+2}, v_{i+3}$, having total depth 5, where the middle segment has depth 3 w.r.t. a Hamiltonian path.}
    \label{fig:aThreePiece}
\end{figure}

\begin{claim} \label{lemma.1}
    Any pair of disjoint Hamiltonian paths on the uniform metric graph $\uniformLineGeomGraph$ has a cut-point if there is a 3-piece with a total depth of less than 10. 
\end{claim}
\begin{proof}
    Let $H_1$ and $H_2$ be the two disjoint Hamiltonian ($s$, $t$)-paths and let $P_n$ be the path graph induced by \uniformLineGeomGraph. First, assume that there is a 3-piece with a total depth of less than 10. Since, by Observation \ref{fact:2}, each segment can have only odd depth, such a piece must have a total depth of 6 or 8; namely, depth 3 w.r.t. both paths or depth 3 w.r.t. one path and depth 5 w.r.t. the other. The former case can be discarded as $H_1$ and $H_2$ are edge-disjoint. The only case realizing a total depth of 6 is if every segment has depth 1 w.r.t. $H_1$ and $H_2$, respectively. But then the edge corresponding to the middle segment of the 3-piece would be present in both $H_1$ and $H_2$. 
    
    The only remaining possibility is for the piece to have a total depth of 8. This occurs only if, without loss of generality, each segment in the piece has a depth of 1 w.r.t. path $H_1$, while w.r.t. path $H_2$, two segments have a depth of 1 and one segment has a depth of 3. By Claim \ref{claim:6}, the middle segment of the piece is an edge of $H_1$. Then, by Claim \ref{claim:7}, one of the extremal segments of the piece---the leftmost or the rightmost---must have depth 3 w.r.t. $H_2$. In either case, there is always a vertex $v$ that is not covered w.r.t. $H_1$ nor $H_2$. This is because the vertices of the middle segment are not covered by $H_1$, and exactly one of these middle vertices is not covered by $H_2$. Therefore, such a vertex $v$ is a cut-point. \hfill$\qed$
\end{proof}

With these ingredients in place, we can now prove the lemma, which we restate for clarity. 

\lemmaNoCutpoint*
\begin{proof}
    Recall that $P_n$ is the path graph formed by the vertex set and the segments of the uniform metric graph $\uniformLineGeomGraph$. 
    We define $m := n - 1$ to be the number of edges in $P_n$, and let $H_1$ and $H_2$ be two edge-disjoint Hamiltonian paths in $\uniformLineGeomGraph$. We analyze three cases based on the number $m$ of edges: 

    \begin{enumerate}[label=(\roman*)]
        \item $m = 3k$ for some $k \in \mathbb{N}$. In this case, we can partition the edges of $P_n$ into $k$ 3-pieces. Because we know that $H_1$ and $H_2$ have no cut-point, by Claim \ref{lemma.1} each 3-piece has a total cost that is greater than or equal to 10, which amounts to a total depth of at least $10k$ for the pair $(H_1, H_2)$. The total cost of the two paths is then at least $\frac{10(n-1)}{3}$, which is greater than $\frac{16(n-1)}{5}$ for $n \geq 1$.

        \item $m = 3k + 1$ for some $k \in \mathbb{N}$. In this case, we can partition $P_n$ into $k$ 3-pieces and one 1-piece. We know that the 1-piece must have a total depth of at least 2. Then, the total depth of all the pieces is at least $10k + 2 = \frac{10(n-1)}{3} - \frac{4}{3}$, where we have used Claim \ref{lemma.1} to bound the total depth of the 3-pieces. Note that $\frac{10(n-1)}{3} - \frac{4}{3} < \frac{16(n-1)}{3}$ iff $n < 11$. Because $k \geq 2$ and $n = m+1 = 3k+2$, we thus have $k = 2$. However, by Observation \ref{claim:1}, no two edge-disjoint paths exist while having a total cost less than $\frac{16(n-1)}{3}$ for $n = 8$. Therefore, we conclude for this case that the total cost of paths $H_1$ and $H_2$ is at least $\frac{16(n-1)}{5}$ for $n \geq 6$. 

        \item $m = 3k + 2$ for some $k \in \mathbb{N}$. In this case, we partition the edges of $P_n$ into $k$ 3-pieces and one 2-piece. Let $s_{i-1}, s_i \in P_n$ be the adjacent segments defining the latter, with $s_{i-1} = (v_{i-1}, v_i)$ and $s_i = (v_i, v_{i+1})$. We now prove that the 2-piece must have a total depth of at least 6. Since the paths $H_1$ and $H_2$ have no cut-point, one of $H_1$ or $H_2$ must cover vertex $v_i$. Then, one of $s_{i-1}$ of $s_i$ must have a total depth greater than 2. Further, because of Observation \ref{fact:2}, such a segment must have a total depth of at least 4. Since the remaining segment has a total depth of at least 2, it follows that the 2-piece defined by $s_{i-1}$ and $s_i$ has a total depth of at least 6.

        Back to our main argument, by Claim \ref{lemma.1}, we then have that the total depth of the edges of $P_n$ is at least $10k+6 = \frac{10(n-1)}{3} - \frac{2}{3}$. This expression is greater than or equal to $\frac{16(n-1)}{3}$ iff $n \geq 6$. Hence, the main claim is also valid in this case. 
    \end{enumerate}
    Having considered all possible cases, the lemma is proved. \hfill$\qed$
\end{proof}

\subsection{Loss ratio analysis of the \texttt{Paths} algorithm} \label{appendix.algorithmProofs}

We first provide a more formal pseudocode description of Algorithm \ref{algo.approxPaths}. For this, we require some notation. Let $X = (x_1, \ldots, x_p)$ and $Y = (y_1, \ldots, y_q)$ be two path graphs. To define their concatenation, we assume that the final vertex of $X$ corresponds to the starting vertex of $Y$, i.e., $x_p = y_1$. The concatenated path, denoted by $X \mathbin\Vert Y$, is defined on the combined vertex set as $X \mathbin\Vert Y = (x_1, \ldots, x_p, y_2, \ldots, y_q)$. We use the notation $X^k$ to denote the concatenation of $k$ copies of $X$. 

\begin{algorithm}[H]
\caption[Caption for LOF]{Algorithm \texttt{Paths}} \label{algo.approxPathsFormal}
\vspace{.5em}
\textbf{Input:} Uniform metric graph \uniformLineGeomGraph \\ 
\textbf{Output:} A pair of disjoint Hamiltonian paths. 
\vspace{.5em}
\begin{algorithmic}[1]
\newcommand\NoDo{\renewcommand\algorithmicdo{}}
\newcommand\ReDo{\renewcommand\algorithmicdo{\textbf{do}}}
\algrenewcommand\algorithmiccomment[1]{\hfill {\color{blue} \(\triangleright\) #1}}

\algdef{SE}[SUBALG]{Indent}{EndIndent}{}{\algorithmicend\ }%
\algtext*{Indent}
\algtext*{EndIndent}

\State Let $(H_i, H'_i)$ be an optimal solution to \disjointSHP{} when the input graph has $6 \leq i \leq 11$ vertices.
\If{$n \leq 11$}
\State Return $(H_n, H'_n)$.
\Else
\State Identify $k$ and $\ell$ such that $(n - 1) = 10k + \ell$, and let $m = 10k$. 
\State Let $(H_{11}, H'_{11})$ be the pair of paths shown in Figure~\ref{fig:OptPathsN11}.
\If{$\ell = 0$}
\State \Return $((H_{11})^k, (H'_{11})^k)$.
\ElsIf{$\ell \in \{5, 6, 7, 8, 9\}$} 
\State \Return $((H_{11})^k \mathbin\Vert H_{\ell+1}, (H'_{11})^k \mathbin\Vert H'_{\ell+1})$.
\Else \Comment{$\ell \in \{1, 2, 3, 4\}$}
\State \begin{varwidth}[t]{0.9\linewidth}
Let $(H, H')$ be the pair of paths for $12 \leq n \leq 15$ shown in Appendix \ref{appendix.optPaths2}.  
\end{varwidth}
\State \Return $((H_{11})^{k-1} \mathbin\Vert H, (H'_{11})^{k-1} \mathbin\Vert H')$.
\EndIf
\EndIf
\end{algorithmic}
\vspace{.5em}
\end{algorithm}

We split Lemma \ref{theorem:LRPathsCombined} into two parts, proving first the loss ratio of the algorithm and then its asymptotic loss ratio.

\begin{restatable}[]{lemma}{lemmaLossRatioPaths} \label{theorem:LRPaths}
    Algorithm \ref{algo.approxPaths} has loss ratio $LR_{\texttt{Paths}} = \frac{13}{7}$. 
\end{restatable}
\begin{proof}
    Let $n$ denote the size of the input graph. We consider two cases. First, if $n \leq 11$, we can determine that the ratio $\frac{\min_{S \in \mathcal{S}(\mathcal{I})} \{\text{cost}(S) \}}{OPT(I)}$ is maximized when $n = 8$, where an optimal pair of disjoint paths has cost 13 (see Figure \ref{fig:OptPathsN8} in Appendix \ref{appendix.optPaths}) and an optimal single path has cost 7. Thus, for $n = 8$ the loss ratio is $\frac{13}{7}$. The algorithm indeed outputs a disjoint pair of paths achieving this ratio when $n = 8$. Second, if $n > 11$, from the proof of Lemma \ref{theorem:aPoDPaths}, we know that the greatest ratio the algorithm can produce is $\frac{23}{13}$, which occurs when $n = 14$. Since $\frac{13}{7} > \frac{23}{13}$, the maximum ratio over all instances is $\frac{13}{7}$, and thus the required loss ratio is obtained. \hfill$\qed$
\end{proof}

\begin{restatable}[]{lemma}{lemmaALossRatioPaths} \label{theorem:aPoDPaths}
    Algorithm \ref{algo.approxPaths} has asymptotic loss ratio $LR^\infty_{\texttt{Paths}} = \frac{8}{5}$. 
\end{restatable}
\begin{proof}
    Note that the cost of the more expensive of the two paths generated by Algorithm \ref{algo.approxPaths} on an input of size $n = 11 + (k-1) \cdot 10 + \ell$, where $k \geq 1$ and $1 \leq \ell \leq 9$, can be expressed as $16k + f(\ell)$, where $f(\ell)$ is never larger than $2 \ell$ (see Table~\ref{tab:fValues}). (The values in this table follow directly from the constructions in Appendices \ref{appendix.optPaths} and \ref{appendix.optPaths2}.) Similarly, the cost of the optimal path is given by $10k + \ell$. Hence, the asymptotic loss ratio is $\lim_{k \rightarrow \infty} \frac{16k + f(\ell)}{10 k + \ell} = \frac{8}{5}$. \hfill$\qed$
\end{proof}

\begin{table}[H]
    \centering
    \begin{tabular}{|c|ccccccccc|}
    \hline
        $\ell$ & 1 & 2 & 3 & 4 & 5 & 6 & 7 & 8 & 9 \\
        \hline
        $f(\ell)$ & 3 & 4 & 7 & 8 & 9 & 10 & 13 & 14 & 15 \\
        \hline
    \end{tabular}
    \caption{Values for $f(\ell)$ of the proof of Lemma \ref{theorem:LRPaths}.}
    \label{tab:fValues}
\end{table}

\section{Deferred Figures} \label{appendix.figures}

\subsection{Some optimal Hamiltonian \texorpdfstring{($s$, $t$)}{(s, t)}-paths for small \texorpdfstring{$n$}{n}} \label{appendix.optPaths}

\begin{figure}[H]
    \centering
    \begin{tikzpicture}
\node [draw=black,circle,inner sep=0pt,minimum size=1pt] (v3) at (0,1.25) {\footnotesize $v_{3}$};
\node [draw=black,circle,inner sep=0pt,minimum size=1pt] (v4) at (0.75,1.25) {\footnotesize $v_{4}$};
\node [draw=black,circle,inner sep=0pt,minimum size=1pt] (v5) at (1.5,1.25) {\footnotesize $v_{5}$};
\node [draw=black,circle,inner sep=0pt,minimum size=1pt] (v6) at (2.25,1.25) {\footnotesize $v_{6}$};
\node [draw=black,circle,inner sep=0pt,minimum size=1pt] (v2) at (-0.75,1.25) {\footnotesize $v_{2}$};
\node [draw=black,circle,inner sep=0pt,minimum size=1pt] (v1) at (-1.5,1.25) {\footnotesize $v_{1}$};
\draw  (v1) edge[bend left=45] (v3);
\draw  (v3) edge (v2);
\draw  (v2) edge[bend right=45] (v4);
\draw  (v4) edge (v5);
\draw  (v5) edge (v6);
\node at (-2.5,1.25) {$H_1: $};

\node [draw=black,circle,inner sep=0pt,minimum size=1pt] (v3) at (0,-0.25) {\footnotesize $v_{3}$};
\node [draw=black,circle,inner sep=0pt,minimum size=1pt] (v4) at (0.75,-0.25) {\footnotesize $v_{4}$};
\node [draw=black,circle,inner sep=0pt,minimum size=1pt] (v5) at (1.5,-0.25) {\footnotesize $v_{5}$};
\node [draw=black,circle,inner sep=0pt,minimum size=1pt] (v6) at (2.25,-0.25) {\footnotesize $v_{6}$};
\node [draw=black,circle,inner sep=0pt,minimum size=1pt] (v2) at (-0.75,-0.25) {\footnotesize $v_{2}$};
\node [draw=black,circle,inner sep=0pt,minimum size=1pt] (v1) at (-1.5,-0.25) {\footnotesize $v_{1}$};
\draw  (v1) edge (v2);
\draw  (v3) edge (v4);
\draw  (v2) edge[bend left=45] (v5);
\draw  (v4) edge[bend left=45] (v6);
\draw  (v5) edge[bend left=45] (v3);
\node at (-2.5,-0.25) {$H_2:$};

\end{tikzpicture}
    \caption{An optimal pair of edge-disjoint Hamiltonian paths with $cost(H_1, H_2) = 9$ for $n = 6$.}
    \label{fig:OptPathsN6}
\end{figure}

\begin{figure}[H]
    \centering
    \begin{tikzpicture}
\node [draw=black,circle,inner sep=0pt,minimum size=1pt] (v3) at (0,1.25) {\footnotesize $v_{3}$};
\node [draw=black,circle,inner sep=0pt,minimum size=1pt] (v4) at (0.75,1.25) {\footnotesize $v_{4}$};
\node [draw=black,circle,inner sep=0pt,minimum size=1pt] (v5) at (1.5,1.25) {\footnotesize $v_{5}$};
\node [draw=black,circle,inner sep=0pt,minimum size=1pt] (v6) at (2.25,1.25) {\footnotesize $v_{6}$};
\node [draw=black,circle,inner sep=0pt,minimum size=1pt] (v2) at (-0.75,1.25) {\footnotesize $v_{2}$};
\node [draw=black,circle,inner sep=0pt,minimum size=1pt] (v1) at (-1.5,1.25) {\footnotesize $v_{1}$};
\node [draw=black,circle,inner sep=0pt,minimum size=1pt] (v7) at (3,1.25) {\footnotesize $v_{7}$};
\draw  (v1) edge[bend left=45] (v3);
\draw  (v3) edge (v4);
\draw  (v2) edge[bend right=45] (v4);
\draw  (v2) edge[bend left=45] (v5);
\draw  (v5) edge (v6);
\draw  (v6) edge (v7);
\node at (-2.5,1.25) {$H_1: $};

\node [draw=black,circle,inner sep=0pt,minimum size=1pt] (v3) at (0,-0.25) {\footnotesize $v_{3}$};
\node [draw=black,circle,inner sep=0pt,minimum size=1pt] (v4) at (0.75,-0.25) {\footnotesize $v_{4}$};
\node [draw=black,circle,inner sep=0pt,minimum size=1pt] (v5) at (1.5,-0.25) {\footnotesize $v_{5}$};
\node [draw=black,circle,inner sep=0pt,minimum size=1pt] (v6) at (2.25,-0.25) {\footnotesize $v_{6}$};
\node [draw=black,circle,inner sep=0pt,minimum size=1pt] (v2) at (-0.75,-0.25) {\footnotesize $v_{2}$};
\node [draw=black,circle,inner sep=0pt,minimum size=1pt] (v1) at (-1.5,-0.25) {\footnotesize $v_{1}$};
\node [draw=black,circle,inner sep=0pt,minimum size=1pt] (v7) at (3,-0.25) {\footnotesize $v_{7}$};
\draw  (v1) edge (v2);
\draw  (v2) edge (v3);
\draw  (v3) edge[bend left=45] (v6);
\draw  (v4) edge[bend right=45] (v6);
\draw  (v4) edge (v5);
\draw  (v5) edge[bend left=45] (v7);
\node at (-2.5,-0.25) {$H_2:$};

\end{tikzpicture}
    \caption{An optimal pair of edge-disjoint Hamiltonian paths with $cost(H_1, H_2) = 10$ for $n = 7$.}
    \label{fig:OptPathsN7}
\end{figure}

\begin{figure}[H]
    \centering
    \begin{tikzpicture}
\node [draw=black,circle,inner sep=0pt,minimum size=1pt] (v3) at (0,1.25) {\footnotesize $v_{3}$};
\node [draw=black,circle,inner sep=0pt,minimum size=1pt] (v4) at (0.75,1.25) {\footnotesize $v_{4}$};
\node [draw=black,circle,inner sep=0pt,minimum size=1pt] (v5) at (1.5,1.25) {\footnotesize $v_{5}$};
\node [draw=black,circle,inner sep=0pt,minimum size=1pt] (v6) at (2.25,1.25) {\footnotesize $v_{6}$};
\node [draw=black,circle,inner sep=0pt,minimum size=1pt] (v2) at (-0.75,1.25) {\footnotesize $v_{2}$};
\node [draw=black,circle,inner sep=0pt,minimum size=1pt] (v1) at (-1.5,1.25) {\footnotesize $v_{1}$};
\node [draw=black,circle,inner sep=0pt,minimum size=1pt] (v7) at (3,1.25) {\footnotesize $v_{7}$};
\node [draw=black,circle,inner sep=0pt,minimum size=1pt] (v8) at (3.75,1.25) {\footnotesize $v_{8}$};
\draw  (v1) edge[bend left=45] (v3);
\draw  (v3) edge (v4);
\draw  (v2) edge[bend right=45] (v4);
\draw  (v2) edge[bend left=45] (v5);
\draw  (v5) edge (v6);
\draw  (v6) edge (v7);
\draw  (v7) edge (v8);
\node at (-2.5,1.25) {$H_1: $};

\node [draw=black,circle,inner sep=0pt,minimum size=1pt] (v3) at (0,-0.25) {\footnotesize $v_{3}$};
\node [draw=black,circle,inner sep=0pt,minimum size=1pt] (v4) at (0.75,-0.25) {\footnotesize $v_{4}$};
\node [draw=black,circle,inner sep=0pt,minimum size=1pt] (v5) at (1.5,-0.25) {\footnotesize $v_{5}$};
\node [draw=black,circle,inner sep=0pt,minimum size=1pt] (v6) at (2.25,-0.25) {\footnotesize $v_{6}$};
\node [draw=black,circle,inner sep=0pt,minimum size=1pt] (v2) at (-0.75,-0.25) {\footnotesize $v_{2}$};
\node [draw=black,circle,inner sep=0pt,minimum size=1pt] (v1) at (-1.5,-0.25) {\footnotesize $v_{1}$};
\node [draw=black,circle,inner sep=0pt,minimum size=1pt] (v7) at (3,-0.25) {\footnotesize $v_{7}$};
\node [draw=black,circle,inner sep=0pt,minimum size=1pt] (v8) at (3.75,-0.25) {\footnotesize $v_{8}$};
\draw  (v1) edge (v2);
\draw  (v2) edge (v3);
\draw  (v3) edge[bend left=45] (v7);
\draw  (v5) edge[bend right=45] (v7);
\draw  (v4) edge (v5);
\draw  (v4) edge[bend left=45] (v6);
\draw  (v6) edge[bend left=45] (v8);
\node at (-2.5,-0.25) {$H_2:$};

\end{tikzpicture}
    \caption{An optimal pair of edge-disjoint Hamiltonian paths with $cost(H_1, H_2) = 13$ for $n = 8$.}
    \label{fig:OptPathsN8}
\end{figure}

\begin{figure}[H]
    \centering
    \begin{tikzpicture}
\node [draw=black,circle,inner sep=0pt,minimum size=1pt] (v3) at (0,1.25) {\footnotesize $v_{3}$};
\node [draw=black,circle,inner sep=0pt,minimum size=1pt] (v4) at (0.75,1.25) {\footnotesize $v_{4}$};
\node [draw=black,circle,inner sep=0pt,minimum size=1pt] (v5) at (1.5,1.25) {\footnotesize $v_{5}$};
\node [draw=black,circle,inner sep=0pt,minimum size=1pt] (v6) at (2.25,1.25) {\footnotesize $v_{6}$};
\node [draw=black,circle,inner sep=0pt,minimum size=1pt] (v2) at (-0.75,1.25) {\footnotesize $v_{2}$};
\node [draw=black,circle,inner sep=0pt,minimum size=1pt] (v1) at (-1.5,1.25) {\footnotesize $v_{1}$};
\node [draw=black,circle,inner sep=0pt,minimum size=1pt] (v7) at (3,1.25) {\footnotesize $v_{7}$};
\node [draw=black,circle,inner sep=0pt,minimum size=1pt] (v8) at (3.75,1.25) {\footnotesize $v_{8}$};
\node [draw=black,circle,inner sep=0pt,minimum size=1pt] (v9) at (4.5,1.25) {\footnotesize $v_{9}$};
\draw  (v1) edge[bend left=45] (v3);
\draw  (v5) edge (v4);
\draw  (v3) edge[bend left=45] (v5);
\draw  (v4) edge[bend left=45] (v2);
\draw  (v2) edge[bend left=45] (v6);
\draw  (v8) edge (v9);
\draw  (v6) edge (v7);
\draw  (v7) edge (v8);
\node at (-2.5,1.25) {$H_1: $};

\node [draw=black,circle,inner sep=0pt,minimum size=1pt] (v3) at (0,-0.25) {\footnotesize $v_{3}$};
\node [draw=black,circle,inner sep=0pt,minimum size=1pt] (v4) at (0.75,-0.25) {\footnotesize $v_{4}$};
\node [draw=black,circle,inner sep=0pt,minimum size=1pt] (v5) at (1.5,-0.25) {\footnotesize $v_{5}$};
\node [draw=black,circle,inner sep=0pt,minimum size=1pt] (v6) at (2.25,-0.25) {\footnotesize $v_{6}$};
\node [draw=black,circle,inner sep=0pt,minimum size=1pt] (v2) at (-0.75,-0.25) {\footnotesize $v_{2}$};
\node [draw=black,circle,inner sep=0pt,minimum size=1pt] (v1) at (-1.5,-0.25) {\footnotesize $v_{1}$};
\node [draw=black,circle,inner sep=0pt,minimum size=1pt] (v7) at (3,-0.25) {\footnotesize $v_{7}$};
\node [draw=black,circle,inner sep=0pt,minimum size=1pt] (v8) at (3.75,-0.25) {\footnotesize $v_{8}$};
\node [draw=black,circle,inner sep=0pt,minimum size=1pt] (v9) at (4.5,-0.25) {\footnotesize $v_{9}$};
\draw  (v1) edge (v2);
\draw  (v2) edge (v3);
\draw  (v3) edge (v4);
\draw  (v4) edge[bend left=45] (v6);
\draw  (v6) edge[bend left=45] (v8);
\draw  (v8) edge[bend left=45] (v5);
\draw  (v5) edge[bend left=45] (v7);
\draw  (v7) edge[bend left=45] (v9);
\node at (-2.5,-0.25) {$H_2:$};

\end{tikzpicture}
    \caption{An optimal pair of edge-disjoint Hamiltonian paths with $cost(H_1, H_2) = 14$ for $n = 9$.}
    \label{fig:OptPathsN9}
\end{figure}

\begin{figure}[H]
    \centering
    \begin{tikzpicture}
\node [draw=black,circle,inner sep=0pt,minimum size=1.25em] (v3) at (0,1.25) {\footnotesize $v_{3}$};
\node [draw=black,circle,inner sep=0pt,minimum size=1.25em] (v4) at (0.75,1.25) {\footnotesize $v_{4}$};
\node [draw=black,circle,inner sep=0pt,minimum size=1.25em] (v5) at (1.5,1.25) {\footnotesize $v_{5}$};
\node [draw=black,circle,inner sep=0pt,minimum size=1.25em] (v6) at (2.25,1.25) {\footnotesize $v_{6}$};
\node [draw=black,circle,inner sep=0pt,minimum size=1.25em] (v2) at (-0.75,1.25) {\footnotesize $v_{2}$};
\node [draw=black,circle,inner sep=0pt,minimum size=1.25em] (v1) at (-1.5,1.25) {\footnotesize $v_{1}$};
\node [draw=black,circle,inner sep=0pt,minimum size=1.25em] (v7) at (3,1.25) {\footnotesize $v_{7}$};
\node [draw=black,circle,inner sep=0pt,minimum size=1.25em] (v8) at (3.75,1.25) {\footnotesize $v_{8}$};
\node [draw=black,circle,inner sep=0pt,minimum size=1.25em] (v9) at (4.5,1.25) {\footnotesize $v_{9}$};
\node [draw=black,circle,inner sep=0pt,minimum size=1em] (v10) at (5.25,1.25) {\footnotesize $v_{10}$};
\draw  (v1) edge[bend left=45] (v3);
\draw  (v3) edge[bend left=45] (v5);
\draw  (v5) edge[bend left=45] (v2);
\draw  (v2) edge[bend left=45] (v4);
\draw  (v4) edge[bend left=45] (v6);
\draw  (v9) edge (v10);
\draw  (v8) edge (v9);
\draw  (v6) edge (v7);
\draw  (v7) edge (v8);
\node at (-2.5,1.25) {$H_1: $};

\node [draw=black,circle,inner sep=0pt,minimum size=1.25em] (v3) at (0,-0.25) {\footnotesize $v_{3}$};
\node [draw=black,circle,inner sep=0pt,minimum size=1.25em] (v4) at (0.75,-0.25) {\footnotesize $v_{4}$};
\node [draw=black,circle,inner sep=0pt,minimum size=1.25em] (v5) at (1.5,-0.25) {\footnotesize $v_{5}$};
\node [draw=black,circle,inner sep=0pt,minimum size=1.25em] (v6) at (2.25,-0.25) {\footnotesize $v_{6}$};
\node [draw=black,circle,inner sep=0pt,minimum size=1.25em] (v2) at (-0.75,-0.25) {\footnotesize $v_{2}$};
\node [draw=black,circle,inner sep=0pt,minimum size=1.25em] (v1) at (-1.5,-0.25) {\footnotesize $v_{1}$};
\node [draw=black,circle,inner sep=0pt,minimum size=1.25em] (v7) at (3,-0.25) {\footnotesize $v_{7}$};
\node [draw=black,circle,inner sep=0pt,minimum size=1.25em] (v8) at (3.75,-0.25) {\footnotesize $v_{8}$};
\node [draw=black,circle,inner sep=0pt,minimum size=1.25em] (v9) at (4.5,-0.25) {\footnotesize $v_{9}$};
\node [draw=black,circle,inner sep=0pt,minimum size=1.25em] (v10) at (5.25,-0.25) {\footnotesize $v_{10}$};
\draw  (v1) edge (v2);
\draw  (v2) edge (v3);
\draw  (v3) edge (v4);
\draw  (v4) edge (v5);
\draw  (v8) edge[bend right=45] (v6);
\draw  (v6) edge[bend right=45] (v9);
\draw  (v8) edge[bend left=45] (v10);
\draw  (v5) edge[bend left=45] (v7);
\draw  (v7) edge[bend left=45] (v9);
\node at (-2.5,-0.25) {$H_2:$};

\end{tikzpicture}
    \caption{An optimal pair of edge-disjoint Hamiltonian paths with $cost(H_1, H_2) = 15$ for $n = 10$.}
    \label{fig:OptPathsN10}
\end{figure}

\begin{figure}[H]
    \centering
    \begin{tikzpicture}
\node [draw=black,circle,inner sep=0pt,minimum size=1.25em] (v3) at (0,1.25) {\footnotesize $v_{3}$};
\node [draw=black,circle,inner sep=0pt,minimum size=1.25em] (v4) at (0.75,1.25) {\footnotesize $v_{4}$};
\node [draw=black,circle,inner sep=0pt,minimum size=1.25em] (v5) at (1.5,1.25) {\footnotesize $v_{5}$};
\node [draw=black,circle,inner sep=0pt,minimum size=1.25em] (v6) at (2.25,1.25) {\footnotesize $v_{6}$};
\node [draw=black,circle,inner sep=0pt,minimum size=1.25em] (v2) at (-0.75,1.25) {\footnotesize $v_{2}$};
\node [draw=black,circle,inner sep=0pt,minimum size=1.25em] (v1) at (-1.5,1.25) {\footnotesize $v_{1}$};
\node [draw=black,circle,inner sep=0pt,minimum size=1.25em] (v7) at (3,1.25) {\footnotesize $v_{7}$};
\node [draw=black,circle,inner sep=0pt,minimum size=1.25em] (v8) at (3.75,1.25) {\footnotesize $v_{8}$};
\node [draw=black,circle,inner sep=0pt,minimum size=1.25em] (v9) at (4.5,1.25) {\footnotesize $v_{9}$};
\node [draw=black,circle,inner sep=0pt,minimum size=1em] (v10) at (5.25,1.25) {\footnotesize $v_{10}$};
\node [draw=black,circle,inner sep=0pt,minimum size=1.25em] (v11) at (6,1.25) {\footnotesize $v_{11}$};
\draw  (v1) edge[bend left=45] (v3);
\draw  (v3) edge (v4);
\draw  (v4) edge[bend right=45] (v2);
\draw  (v2) edge[bend right=45] (v5);
\draw  (v5) edge (v6);
\draw  (v6) edge (v7);
\draw  (v7) edge[bend left=45] (v9);
\draw  (v9) edge (v8);
\draw  (v8) edge[bend left=45] (v10);
\draw  (v10) edge (v11);
\node at (-2.5,1.25) {$H_1: $};

\node [draw=black,circle,inner sep=0pt,minimum size=1.25em] (v3) at (0,-0.25) {\footnotesize $v_{3}$};
\node [draw=black,circle,inner sep=0pt,minimum size=1.25em] (v4) at (0.75,-0.25) {\footnotesize $v_{4}$};
\node [draw=black,circle,inner sep=0pt,minimum size=1.25em] (v5) at (1.5,-0.25) {\footnotesize $v_{5}$};
\node [draw=black,circle,inner sep=0pt,minimum size=1.25em] (v6) at (2.25,-0.25) {\footnotesize $v_{6}$};
\node [draw=black,circle,inner sep=0pt,minimum size=1.25em] (v2) at (-0.75,-0.25) {\footnotesize $v_{2}$};
\node [draw=black,circle,inner sep=0pt,minimum size=1.25em] (v1) at (-1.5,-0.25) {\footnotesize $v_{1}$};
\node [draw=black,circle,inner sep=0pt,minimum size=1.25em] (v7) at (3,-0.25) {\footnotesize $v_{7}$};
\node [draw=black,circle,inner sep=0pt,minimum size=1.25em] (v8) at (3.75,-0.25) {\footnotesize $v_{8}$};
\node [draw=black,circle,inner sep=0pt,minimum size=1.25em] (v9) at (4.5,-0.25) {\footnotesize $v_{9}$};
\node [draw=black,circle,inner sep=0pt,minimum size=1.25em] (v10) at (5.25,-0.25) {\footnotesize $v_{10}$};
\node [draw=black,circle,inner sep=0pt,minimum size=1.25em] (v11) at (6,-0.25) {\footnotesize $v_{11}$};
\draw  (v1) edge (v2);
\draw  (v2) edge (v3);
\draw  (v3) edge[bend left=45] (v5);
\draw  (v5) edge (v4);
\draw  (v4) edge[bend left=45] (v6);
\draw  (v6) edge[bend left=45] (v8);
\draw  (v8) edge (v7);
\draw  (v7) edge[bend right=45] (v10);
\draw  (v10) edge (v9);
\draw  (v9) edge[bend left=45] (v11);
\node at (-2.5,-0.25) {$H_2:$};

\end{tikzpicture}
    \caption{An optimal pair of edge-disjoint Hamiltonian paths with $cost(H_1, H_2) = 16$ 
    for $n = 11$. (Same as Figure \ref{fig:OptPathsN11}.)}
    \label{fig:OptPathsN11_appendix}
\end{figure}

\subsection{Precomputed Hamiltonian \texorpdfstring{($s$, $t$)}{(s, t)}-paths for \texorpdfstring{$12 \leq n \leq 15$}{13 <= n <= 15}} \label{appendix.optPaths2}

Red edges indicate modified or newly added edges compared to the optimal solution of Figure~\ref{fig:OptPathsN11} for $n = 11$, while gray nodes highlight the newly added vertices.

\begin{figure}
 \centering
     \begin{tikzpicture}[scale=1.0, transform shape]
\node [draw=black,circle,inner sep=0pt,minimum size=1.25em] (v3) at (0,1.25) {\footnotesize $v_{3}$};
\node [draw=black,circle,inner sep=0pt,minimum size=1.25em] (v4) at (0.75,1.25) {\footnotesize $v_{4}$};
\node [draw=black,circle,inner sep=0pt,minimum size=1.25em] (v5) at (1.5,1.25) {\footnotesize $v_{5}$};
\node [draw=black,circle,inner sep=0pt,minimum size=1.25em] (v6) at (2.25,1.25) {\footnotesize $v_{6}$};
\node [draw=black,circle,inner sep=0pt,minimum size=1.25em] (v2) at (-0.75,1.25) {\footnotesize $v_{2}$};
\node [draw=black,circle,inner sep=0pt,minimum size=1.25em] (v1) at (-1.5,1.25) {\footnotesize $v_{1}$};
\node [draw=black,circle,inner sep=0pt,minimum size=1.25em] (v7) at (3,1.25) {\footnotesize $v_{7}$};
\node [draw=black,circle,inner sep=0pt,minimum size=1.25em] (v8) at (3.75,1.25) {\footnotesize $v_{8}$};
\node [draw=black,circle,inner sep=0pt,minimum size=1.25em] (v9) at (4.5,1.25) {\footnotesize $v_{9}$};
\node [draw=black,circle,inner sep=0pt,minimum size=1em] (v10) at (5.25,1.25) {\footnotesize $v_{10}$};
\node [draw,circle,inner sep=0pt,minimum size=1.25em] (v11) at (6,1.25) {\footnotesize $v_{11}$};
\node [draw=red,circle,inner sep=0pt,minimum size=1.25em, fill=black!20] (v12) at (6.75,1.25) {\footnotesize $v_{12}$};
\draw  (v1) edge[bend left=45] (v3);
\draw  (v3) edge (v4);
\draw  (v4) edge[bend right=45] (v2);
\draw  (v2) edge[bend right=45] (v5);
\draw  (v5) edge (v6);
\draw  (v6) edge (v7);
\draw  (v7) edge[bend left=45] (v9);
\draw  (v9) edge (v8);
\draw[red]  (v8) edge[bend right=45] (v11);
\draw  (v10) edge (v11);
\draw[red]  (v10) edge[bend left=45] (v12);
\node at (-2.5,1.25) {$H_1: $};

\node [draw=black,circle,inner sep=0pt,minimum size=1.25em] (v3) at (0,-0.25) {\footnotesize $v_{3}$};
\node [draw=black,circle,inner sep=0pt,minimum size=1.25em] (v4) at (0.75,-0.25) {\footnotesize $v_{4}$};
\node [draw=black,circle,inner sep=0pt,minimum size=1.25em] (v5) at (1.5,-0.25) {\footnotesize $v_{5}$};
\node [draw=black,circle,inner sep=0pt,minimum size=1.25em] (v6) at (2.25,-0.25) {\footnotesize $v_{6}$};
\node [draw=black,circle,inner sep=0pt,minimum size=1.25em] (v2) at (-0.75,-0.25) {\footnotesize $v_{2}$};
\node [draw=black,circle,inner sep=0pt,minimum size=1.25em] (v1) at (-1.5,-0.25) {\footnotesize $v_{1}$};
\node [draw=black,circle,inner sep=0pt,minimum size=1.25em] (v7) at (3,-0.25) {\footnotesize $v_{7}$};
\node [draw=black,circle,inner sep=0pt,minimum size=1.25em] (v8) at (3.75,-0.25) {\footnotesize $v_{8}$};
\node [draw=black,circle,inner sep=0pt,minimum size=1.25em] (v9) at (4.5,-0.25) {\footnotesize $v_{9}$};
\node [draw=black,circle,inner sep=0pt,minimum size=1.25em] (v10) at (5.25,-0.25) {\footnotesize $v_{10}$};
\node [draw=black,circle,inner sep=0pt,minimum size=1.25em] (v11) at (6,-0.25) {\footnotesize $v_{11}$};
\node [draw=red,circle,inner sep=0pt,minimum size=1.25em,fill=black!20] (v12) at (6.75,-0.25) {\footnotesize $v_{12}$};
\draw  (v1) edge (v2);
\draw  (v2) edge (v3);
\draw  (v3) edge[bend left=45] (v5);
\draw  (v5) edge (v4);
\draw  (v4) edge[bend left=45] (v6);
\draw  (v6) edge[bend left=45] (v8);
\draw  (v8) edge (v7);
\draw  (v7) edge[bend right=45] (v10);
\draw  (v10) edge (v9);
\draw  (v9) edge[bend left=45] (v11);
\draw[red]  (v11) edge (v12);
\node at (-2.5,-0.25) {$H_2:$};

\end{tikzpicture}
    \caption{Pair of edge-disjoint Hamiltonian paths for $n=12$ with $cost(H_1, H_2) = 19$.}
    \label{fig:algoAddNodes12_appendix}
 \end{figure}
 
 \begin{figure}
 \centering
     \begin{tikzpicture}[scale=1.0, transform shape]
\node [draw=black,circle,inner sep=0pt,minimum size=1.25em] (v3) at (0,1.25) {\footnotesize $v_{3}$};
\node [draw=black,circle,inner sep=0pt,minimum size=1.25em] (v4) at (0.75,1.25) {\footnotesize $v_{4}$};
\node [draw=black,circle,inner sep=0pt,minimum size=1.25em] (v5) at (1.5,1.25) {\footnotesize $v_{5}$};
\node [draw=black,circle,inner sep=0pt,minimum size=1.25em] (v6) at (2.25,1.25) {\footnotesize $v_{6}$};
\node [draw=black,circle,inner sep=0pt,minimum size=1.25em] (v2) at (-0.75,1.25) {\footnotesize $v_{2}$};
\node [draw=black,circle,inner sep=0pt,minimum size=1.25em] (v1) at (-1.5,1.25) {\footnotesize $v_{1}$};
\node [draw=black,circle,inner sep=0pt,minimum size=1.25em] (v7) at (3,1.25) {\footnotesize $v_{7}$};
\node [draw=black,circle,inner sep=0pt,minimum size=1.25em] (v8) at (3.75,1.25) {\footnotesize $v_{8}$};
\node [draw=black,circle,inner sep=0pt,minimum size=1.25em] (v9) at (4.5,1.25) {\footnotesize $v_{9}$};
\node [draw=black,circle,inner sep=0pt,minimum size=1em] (v10) at (5.25,1.25) {\footnotesize $v_{10}$};
\node [draw,circle,inner sep=0pt,minimum size=1.25em] (v11) at (6,1.25) {\footnotesize $v_{11}$};
\node [draw=red,circle,inner sep=0pt,minimum size=1.25em,fill=black!20] (v12) at (6.75,1.25) {\footnotesize $v_{12}$};
\node [draw=red,circle,inner sep=0pt,minimum size=1.25em,fill=black!20] (v13) at (7.5,1.25) {\footnotesize $v_{13}$};
\draw  (v1) edge[bend left=45] (v3);
\draw  (v3) edge (v4);
\draw  (v4) edge[bend right=45] (v2);
\draw  (v2) edge[bend right=45] (v5);
\draw  (v5) edge (v6);
\draw  (v6) edge (v7);
\draw  (v7) edge[bend left=45] (v9);
\draw  (v9) edge (v8);
\draw[red]  (v8) edge[bend right=45] (v11);
\draw  (v10) edge (v11);
\draw[red]  (v10) edge[bend left=45] (v12);
\draw[red]  (v12) edge (v13);
\node at (-2.5,1.25) {$H_1: $};

\node [draw=black,circle,inner sep=0pt,minimum size=1.25em] (v3) at (0,-0.25) {\footnotesize $v_{3}$};
\node [draw=black,circle,inner sep=0pt,minimum size=1.25em] (v4) at (0.75,-0.25) {\footnotesize $v_{4}$};
\node [draw=black,circle,inner sep=0pt,minimum size=1.25em] (v5) at (1.5,-0.25) {\footnotesize $v_{5}$};
\node [draw=black,circle,inner sep=0pt,minimum size=1.25em] (v6) at (2.25,-0.25) {\footnotesize $v_{6}$};
\node [draw=black,circle,inner sep=0pt,minimum size=1.25em] (v2) at (-0.75,-0.25) {\footnotesize $v_{2}$};
\node [draw=black,circle,inner sep=0pt,minimum size=1.25em] (v1) at (-1.5,-0.25) {\footnotesize $v_{1}$};
\node [draw=black,circle,inner sep=0pt,minimum size=1.25em] (v7) at (3,-0.25) {\footnotesize $v_{7}$};
\node [draw=black,circle,inner sep=0pt,minimum size=1.25em] (v8) at (3.75,-0.25) {\footnotesize $v_{8}$};
\node [draw=black,circle,inner sep=0pt,minimum size=1.25em] (v9) at (4.5,-0.25) {\footnotesize $v_{9}$};
\node [draw=black,circle,inner sep=0pt,minimum size=1.25em] (v10) at (5.25,-0.25) {\footnotesize $v_{10}$};
\node [draw=black,circle,inner sep=0pt,minimum size=1.25em] (v11) at (6,-0.25) {\footnotesize $v_{11}$};
\node [draw=red,circle,inner sep=0pt,minimum size=1.25em,fill=black!20] (v12) at (6.75,-0.25) {\footnotesize $v_{12}$};
\node [draw=red,circle,inner sep=0pt,minimum size=1.25em,fill=black!20] (v13) at (7.5,-0.25) {\footnotesize $v_{13}$};
\draw  (v1) edge (v2);
\draw  (v2) edge (v3);
\draw  (v3) edge[bend left=45] (v5);
\draw  (v5) edge (v4);
\draw  (v4) edge[bend left=45] (v6);
\draw  (v6) edge[bend left=45] (v8);
\draw  (v8) edge (v7);
\draw  (v7) edge[bend right=45] (v10);
\draw  (v10) edge (v9);
\draw[red]  (v9) edge[bend left=45] (v12);
\draw[red]  (v11) edge (v12);
\draw[red]  (v11) edge[bend right=45] (v13);
\node at (-2.5,-0.25) {$H_2:$};

\end{tikzpicture}
    \caption{Pair of edge-disjoint Hamiltonian paths for $n=13$ with $cost(H_1, H_2) = 20$.}
    \label{fig:algoAddNode13}
 \end{figure}

 \begin{figure}
 \centering
     \begin{tikzpicture}[scale=1.0, transform shape]
\node [draw=black,circle,inner sep=0pt,minimum size=1.25em] (v3) at (0,1.25) {\footnotesize $v_{3}$};
\node [draw=black,circle,inner sep=0pt,minimum size=1.25em] (v4) at (0.75,1.25) {\footnotesize $v_{4}$};
\node [draw=black,circle,inner sep=0pt,minimum size=1.25em] (v5) at (1.5,1.25) {\footnotesize $v_{5}$};
\node [draw=black,circle,inner sep=0pt,minimum size=1.25em] (v6) at (2.25,1.25) {\footnotesize $v_{6}$};
\node [draw=black,circle,inner sep=0pt,minimum size=1.25em] (v2) at (-0.75,1.25) {\footnotesize $v_{2}$};
\node [draw=black,circle,inner sep=0pt,minimum size=1.25em] (v1) at (-1.5,1.25) {\footnotesize $v_{1}$};
\node [draw=black,circle,inner sep=0pt,minimum size=1.25em] (v7) at (3,1.25) {\footnotesize $v_{7}$};
\node [draw=black,circle,inner sep=0pt,minimum size=1.25em] (v8) at (3.75,1.25) {\footnotesize $v_{8}$};
\node [draw=black,circle,inner sep=0pt,minimum size=1.25em] (v9) at (4.5,1.25) {\footnotesize $v_{9}$};
\node [draw=black,circle,inner sep=0pt,minimum size=1em] (v10) at (5.25,1.25) {\footnotesize $v_{10}$};
\node [draw,circle,inner sep=0pt,minimum size=1.25em] (v11) at (6,1.25) {\footnotesize $v_{11}$};
\node [draw=red,circle,inner sep=0pt,minimum size=1.25em,fill=black!20] (v12) at (6.75,1.25) {\footnotesize $v_{12}$};
\node [draw=red,circle,inner sep=0pt,minimum size=1.25em,fill=black!20] (v13) at (7.5,1.25) {\footnotesize $v_{13}$};
\node [draw=red,circle,inner sep=0pt,minimum size=1.25em,fill=black!20] (v14) at (8.25,1.25) {\footnotesize $v_{14}$};
\draw  (v1) edge[bend left=45] (v3);
\draw  (v3) edge (v4);
\draw  (v4) edge[bend right=45] (v2);
\draw  (v2) edge[bend right=45] (v5);
\draw  (v5) edge (v6);
\draw  (v6) edge (v7);
\draw  (v7) edge[bend left=45] (v9);
\draw  (v9) edge (v8);
\draw[red]  (v8) edge[bend right=45] (v11);
\draw  (v10) edge (v11);
\draw[red]  (v10) edge[bend left=45] (v13);
\draw[red]  (v12) edge (v13);
\draw[red]  (v12) edge[bend right=45] (v14);
\node at (-2.5,1.25) {$H_1: $};

\node [draw=black,circle,inner sep=0pt,minimum size=1.25em] (v3) at (0,-0.25) {\footnotesize $v_{3}$};
\node [draw=black,circle,inner sep=0pt,minimum size=1.25em] (v4) at (0.75,-0.25) {\footnotesize $v_{4}$};
\node [draw=black,circle,inner sep=0pt,minimum size=1.25em] (v5) at (1.5,-0.25) {\footnotesize $v_{5}$};
\node [draw=black,circle,inner sep=0pt,minimum size=1.25em] (v6) at (2.25,-0.25) {\footnotesize $v_{6}$};
\node [draw=black,circle,inner sep=0pt,minimum size=1.25em] (v2) at (-0.75,-0.25) {\footnotesize $v_{2}$};
\node [draw=black,circle,inner sep=0pt,minimum size=1.25em] (v1) at (-1.5,-0.25) {\footnotesize $v_{1}$};
\node [draw=black,circle,inner sep=0pt,minimum size=1.25em] (v7) at (3,-0.25) {\footnotesize $v_{7}$};
\node [draw=black,circle,inner sep=0pt,minimum size=1.25em] (v8) at (3.75,-0.25) {\footnotesize $v_{8}$};
\node [draw=black,circle,inner sep=0pt,minimum size=1.25em] (v9) at (4.5,-0.25) {\footnotesize $v_{9}$};
\node [draw=black,circle,inner sep=0pt,minimum size=1.25em] (v10) at (5.25,-0.25) {\footnotesize $v_{10}$};
\node [draw=black,circle,inner sep=0pt,minimum size=1.25em] (v11) at (6,-0.25) {\footnotesize $v_{11}$};
\node [draw=red,circle,inner sep=0pt,minimum size=1.25em,fill=black!20] (v12) at (6.75,-0.25) {\footnotesize $v_{12}$};
\node [draw=red,circle,inner sep=0pt,minimum size=1.25em,fill=black!20] (v13) at (7.5,-0.25) {\footnotesize $v_{13}$};
\node [draw=red,circle,inner sep=0pt,minimum size=1.25em,fill=black!20] (v14) at (8.25,-0.25) {\footnotesize $v_{14}$};
\draw  (v1) edge (v2);
\draw  (v2) edge (v3);
\draw  (v3) edge[bend left=45] (v5);
\draw  (v5) edge (v4);
\draw  (v4) edge[bend left=45] (v6);
\draw  (v6) edge[bend left=45] (v8);
\draw  (v8) edge (v7);
\draw  (v7) edge[bend right=45] (v10);
\draw  (v10) edge (v9);
\draw[red]  (v9) edge[bend left=45] (v12);
\draw[red]  (v11) edge (v12);
\draw[red]  (v11) edge[bend right=45] (v13);
\draw[red]  (v13) edge (v14);
\node at (-2.5,-0.25) {$H_2:$};
\end{tikzpicture}
    \caption{Pair of edge-disjoint Hamiltonian paths for $n=14$ with $cost(H_1, H_2) = 23$.}    \label{fig:algoAddNode14}
 \end{figure}

 \begin{figure}
 \centering
     \begin{tikzpicture}[scale=1.0, transform shape]
\node [draw=black,circle,inner sep=0pt,minimum size=1.25em] (v3) at (0,1.25) {\footnotesize $v_{3}$};
\node [draw=black,circle,inner sep=0pt,minimum size=1.25em] (v4) at (0.75,1.25) {\footnotesize $v_{4}$};
\node [draw=black,circle,inner sep=0pt,minimum size=1.25em] (v5) at (1.5,1.25) {\footnotesize $v_{5}$};
\node [draw=black,circle,inner sep=0pt,minimum size=1.25em] (v6) at (2.25,1.25) {\footnotesize $v_{6}$};
\node [draw=black,circle,inner sep=0pt,minimum size=1.25em] (v2) at (-0.75,1.25) {\footnotesize $v_{2}$};
\node [draw=black,circle,inner sep=0pt,minimum size=1.25em] (v1) at (-1.5,1.25) {\footnotesize $v_{1}$};
\node [draw=black,circle,inner sep=0pt,minimum size=1.25em] (v7) at (3,1.25) {\footnotesize $v_{7}$};
\node [draw=black,circle,inner sep=0pt,minimum size=1.25em] (v8) at (3.75,1.25) {\footnotesize $v_{8}$};
\node [draw=black,circle,inner sep=0pt,minimum size=1.25em] (v9) at (4.5,1.25) {\footnotesize $v_{9}$};
\node [draw=black,circle,inner sep=0pt,minimum size=1em] (v10) at (5.25,1.25) {\footnotesize $v_{10}$};
\node [draw,circle,inner sep=0pt,minimum size=1.25em] (v11) at (6,1.25) {\footnotesize $v_{11}$};
\node [draw=red,circle,inner sep=0pt,minimum size=1.25em,fill=black!20] (v12) at (6.75,1.25) {\footnotesize $v_{12}$};
\node [draw=red,circle,inner sep=0pt,minimum size=1.25em,fill=black!20] (v13) at (7.5,1.25) {\footnotesize $v_{13}$};
\node [draw=red,circle,inner sep=0pt,minimum size=1.25em,fill=black!20] (v14) at (8.25,1.25) {\footnotesize $v_{14}$};
\node [draw=red,circle,inner sep=0pt,minimum size=1.25em,fill=black!20] (v15) at (9,1.25) {\footnotesize $v_{15}$};
\draw  (v1) edge[bend left=45] (v3);
\draw  (v3) edge (v4);
\draw  (v4) edge[bend right=45] (v2);
\draw  (v2) edge[bend right=45] (v5);
\draw  (v5) edge (v6);
\draw  (v6) edge (v7);
\draw  (v7) edge[bend left=45] (v9);
\draw  (v9) edge (v8);
\draw[red]  (v8) edge[bend right=45] (v11);
\draw  (v10) edge (v11);
\draw[red]  (v10) edge[bend left=45] (v13);
\draw[red]  (v12) edge (v13);
\draw[red]  (v12) edge[bend right=45] (v14);
\draw[red]  (v14) edge (v15);
\node at (-2.5,1.25) {$H_1: $};

\node [draw=black,circle,inner sep=0pt,minimum size=1.25em] (v3) at (0,-0.25) {\footnotesize $v_{3}$};
\node [draw=black,circle,inner sep=0pt,minimum size=1.25em] (v4) at (0.75,-0.25) {\footnotesize $v_{4}$};
\node [draw=black,circle,inner sep=0pt,minimum size=1.25em] (v5) at (1.5,-0.25) {\footnotesize $v_{5}$};
\node [draw=black,circle,inner sep=0pt,minimum size=1.25em] (v6) at (2.25,-0.25) {\footnotesize $v_{6}$};
\node [draw=black,circle,inner sep=0pt,minimum size=1.25em] (v2) at (-0.75,-0.25) {\footnotesize $v_{2}$};
\node [draw=black,circle,inner sep=0pt,minimum size=1.25em] (v1) at (-1.5,-0.25) {\footnotesize $v_{1}$};
\node [draw=black,circle,inner sep=0pt,minimum size=1.25em] (v7) at (3,-0.25) {\footnotesize $v_{7}$};
\node [draw=black,circle,inner sep=0pt,minimum size=1.25em] (v8) at (3.75,-0.25) {\footnotesize $v_{8}$};
\node [draw=black,circle,inner sep=0pt,minimum size=1.25em] (v9) at (4.5,-0.25) {\footnotesize $v_{9}$};
\node [draw=black,circle,inner sep=0pt,minimum size=1.25em] (v10) at (5.25,-0.25) {\footnotesize $v_{10}$};
\node [draw=black,circle,inner sep=0pt,minimum size=1.25em] (v11) at (6,-0.25) {\footnotesize $v_{11}$};
\node [draw=red,circle,inner sep=0pt,minimum size=1.25em,fill=black!20] (v12) at (6.75,-0.25) {\footnotesize $v_{12}$};
\node [draw=red,circle,inner sep=0pt,minimum size=1.25em,fill=black!20] (v13) at (7.5,-0.25) {\footnotesize $v_{13}$};
\node [draw=red,circle,inner sep=0pt,minimum size=1.25em,fill=black!20] (v14) at (8.25,-0.25) {\footnotesize $v_{14}$};
\node [draw=red,circle,inner sep=0pt,minimum size=1.25em,fill=black!20] (v15) at (9,-0.25) {\footnotesize $v_{15}$};
\draw  (v1) edge (v2);
\draw  (v2) edge (v3);
\draw  (v3) edge[bend left=45] (v5);
\draw  (v5) edge (v4);
\draw  (v4) edge[bend left=45] (v6);
\draw  (v6) edge[bend left=45] (v8);
\draw  (v8) edge (v7);
\draw  (v7) edge[bend right=45] (v10);
\draw  (v10) edge (v9);
\draw[red]  (v9) edge[bend left=45] (v12);
\draw[red]  (v11) edge (v12);
\draw[red]  (v11) edge[bend right=45] (v14);
\draw[red]  (v13) edge (v14);
\draw[red]  (v13) edge[bend left=45] (v15);
\node at (-2.5,-0.25) {$H_2:$};
\end{tikzpicture}
    \caption{Pair of edge-disjoint Hamiltonian paths for $n=15$ with $cost(H_1, H_2) = 24$.}
    \label{fig:algoAddNode15}
 \end{figure}

\subsection{Pairs of tours generated by Algorithm \ref{algo.approxNaive}} \label{appendix.toursAlgo}
\begin{figure}[H]
  \begin{minipage}[t]{0.45\textwidth} 
    \begin{tikzpicture}[scale=0.8]
\node [draw=black,circle,inner sep=0pt,minimum size=1pt] (v3) at (0.75,0.5) {\footnotesize $v_{3}$};
\node [draw=black,circle,inner sep=0pt,minimum size=1pt] (v4) at (1.5,0.5) {\footnotesize $v_{4}$};
\node [draw=black,circle,inner sep=0pt,minimum size=1pt] (v5) at (2.25,0.5) {\footnotesize $v_{5}$};
\node [draw=black,circle,inner sep=0pt,minimum size=1pt] (v2) at (0,0.5) {\footnotesize $v_{2}$};
\node [draw=black,circle,inner sep=0pt,minimum size=1pt] (v1) at (-0.75,0.5) {\footnotesize $v_{1}$};
\node [draw=black,circle,inner sep=0pt,minimum size=1pt] (v7) at (3,0.5) {\footnotesize $v_{6}$};
\node [draw=black,circle,inner sep=0pt,minimum size=1pt] (v8) at (3.75,0.5) {\footnotesize $v_{7}$};

\draw  (v1) edge[bend right=45] (v3);
\draw  (v3) edge[bend right=45] (v5);
\draw  (v5) edge[bend right=45] (v8);
\draw  (v2) edge[bend left=45] (v4);
\draw  (v4) edge[bend left=45] (v7);
\node (t1) at (-2.5,0.5) {$T_2: $};

\node (aux2T1) at (-1.25,0.75) {};
\node (aux5T1) at (4.375,0.875) {};
\node (aux7T1) at (4.375,0.125) {};
\node (aux9T1) at (-1.25,0) {};
\draw (v2) edge[out=130,  in=30] (aux2T1);
\draw (v8) edge[out=50,  in=185] (aux5T1);
\draw (v7) edge[out=320,  in=200] (aux7T1);
\draw (v1) edge[out=230,  in=40] (aux9T1);

\node [draw=black,circle,inner sep=0pt,minimum size=1pt] (v3) at (0.75,2) {\footnotesize $v_{3}$};
\node [draw=black,circle,inner sep=0pt,minimum size=1pt] (v4) at (1.5,2) {\footnotesize $v_{4}$};
\node [draw=black,circle,inner sep=0pt,minimum size=1pt] (v5) at (2.25,2) {\footnotesize $v_{5}$};
\node [draw=black,circle,inner sep=0pt,minimum size=1pt] (v6) at (3,2) {\footnotesize $v_{6}$};
\node [draw=black,circle,inner sep=0pt,minimum size=1pt] (v2) at (0,2) {\footnotesize $v_{2}$};
\node [draw=black,circle,inner sep=0pt,minimum size=1pt] (v1) at (-0.75,2) {\footnotesize $v_{1}$};
\node [draw=black,circle,inner sep=0pt,minimum size=1pt] (v7) at (3.75,2) {\footnotesize $v_{7}$};

\node (aux1T2) at (-1.25,2) {};
\node (aux2T2) at (4.25,2) {};

\draw (v1) edge (v2);
\draw (v2) edge (v3);
\draw (v3) edge (v4);
\draw (v4) edge (v5);
\draw (v5) edge (v6);
\draw (v6) edge (v7);
\draw (aux1T2) edge (v1);
\draw (aux2T2) edge (v7);
\node at (-2.5,2) {$T_1:$};
\end{tikzpicture}
    \caption{Pair of tours $(T_1, T_2)$ generated by Algorithm \ref{algo.approxNaive} for $n = 7$ (odd case).}
    \label{fig:AlgoToursN7}
  \end{minipage}%
  \hfill
  \begin{minipage}[t]{0.45\textwidth}
    \begin{tikzpicture}[scale=0.8]
\node [draw=black,circle,inner sep=0pt,minimum size=1pt] (v3) at (0.75,1.25) {\footnotesize $v_{3}$};
\node [draw=black,circle,inner sep=0pt,minimum size=1pt] (v4) at (1.5,1.25) {\footnotesize $v_{4}$};
\node [draw=black,circle,inner sep=0pt,minimum size=1pt] (v5) at (2.25,1.25) {\footnotesize $v_{5}$};
\node [draw=black,circle,inner sep=0pt,minimum size=1pt] (v6) at (3,1.25) {\footnotesize $v_{6}$};
\node [draw=black,circle,inner sep=0pt,minimum size=1pt] (v2) at (0,1.25) {\footnotesize $v_{2}$};
\node [draw=black,circle,inner sep=0pt,minimum size=1pt] (v1) at (-0.75,1.25) {\footnotesize $v_{1}$};
\node [draw=black,circle,inner sep=0pt,minimum size=1pt] (v7) at (3.75,1.25) {\footnotesize $v_{7}$};
\node [draw=black,circle,inner sep=0pt,minimum size=1pt] (v8) at (4.5,1.25) {\footnotesize $v_{8}$};
\draw  (v3) edge (v4);
\draw  (v2) edge (v3);
\draw  (v4) edge (v5);
\draw  (v5) edge (v6);
\draw  (v6) edge (v7);
\node (t1) at (-2.5,1.25) {$T_1: $};
\node (aux1T1) at (-1.25,1.25) {};
\node (aux2T1) at (-1.25,1.5) {};
\node (aux4T1) at (5.125,1.25) {};
\node (aux5T1) at (5.125,1.625) {};
\node (aux7T1) at (5.125,0.875) {};
\node (aux9T1) at (-1.25,0.75) {};

\draw (v1) edge (aux1T1);
\draw (v2) edge[out=130,  in=30] (aux2T1);

\draw (v8) edge (aux4T1);
\draw (v8) edge[out=50,  in=185] (aux5T1);

\draw (v7) edge[out=320,  in=200] (aux7T1);

\draw (v1) edge[out=230,  in=40] (aux9T1);

\node [draw=black,circle,inner sep=0pt,minimum size=1pt] (v3) at (0.75,-0.25) {\footnotesize $v_{3}$};
\node [draw=black,circle,inner sep=0pt,minimum size=1pt] (v4) at (1.5,-0.25) {\footnotesize $v_{4}$};
\node [draw=black,circle,inner sep=0pt,minimum size=1pt] (v5) at (2.25,-0.25) {\footnotesize $v_{5}$};
\node [draw=black,circle,inner sep=0pt,minimum size=1pt] (v6) at (3,-0.25) {\footnotesize $v_{6}$};
\node [draw=black,circle,inner sep=0pt,minimum size=1pt] (v2) at (0,-0.25) {\footnotesize $v_{2}$};
\node [draw=black,circle,inner sep=0pt,minimum size=1pt] (v1) at (-0.75,-0.25) {\footnotesize $v_{1}$};
\node [draw=black,circle,inner sep=0pt,minimum size=1pt] (v7) at (3.75,-0.25) {\footnotesize $v_{7}$};
\node [draw=black,circle,inner sep=0pt,minimum size=1pt] (v8) at (4.5,-0.25) {\footnotesize $v_{8}$};
\draw  (v1) edge[bend left=45] (v3);
\draw (v1) edge (v2);
\draw  (v2) edge[bend right=45] (v4);
\draw  (v3) edge[bend left=45] (v5);
\draw  (v5) edge[bend left=45] (v7);
\draw (v7) edge (v8);
\draw  (v4) edge[bend right=45] (v6);
\draw  (v6) edge[bend right=45] (v8);
\node at (-2.5,-0.25) {$T_2:$};
\end{tikzpicture}
    \caption{Pair of tours $(T_1, T_2)$ generated by Algorithm \ref{algo.approxNaive} for $n = 8$ (even case).}
    \label{fig:AlgoToursN8}
  \end{minipage}
\end{figure}

\section{Proofs of Section \ref{sec.2}} \label{appendix.sec3}
\subsection{Proof of Lemma \ref{lemma:3}} \label{appendix.sec3.Lemma.6}
\begin{proof}
    For the sake of contradiction, suppose that there is a tour $T$ such that there are segments of $C_n$ with even 
    depth and segments with odd depth w.r.t. $T$. Then, there must exist two segments $s_1, s_2 \in C_n$, one with even depth and another with odd depth, that are adjacent. W.l.o.g., let $s_1 = (u, v)$ and $s_2 = (v, w)$ have even and odd depth, respectively. Consider the shared vertex $v$. Since $T$ is a cycle, $v$ must be an endpoint of exactly two edges $e, f \in T$. There are two different ways in which $e$ and $f$ contribute towards the depths of segments $s_1$ and $s_2$: (1) $e$ and $f$ cover exclusively one of $s_1$ and $s_2$, or (2) $e$ and $f$ both cover segment $s_1$, or both cover segment $s_2$. We now prove that both cases lead to a contradiction. 

    Consider case (1). W.l.o.g., assume that the edge $e$ covers segment $s_1$ and edge $f$ covers segment $s_2$. Since $v$ is an endpoint of both $e$ and $f$, it cannot be connected to any more edges in tour $T$ and thus, any other edge of $T$ that covers one of $s_1$ or $s_2$, must also cover the other. Therefore, we get that $\mu_T(s_1) = \mu_T(s_2)$. This is a contradiction, since we assumed that $\mu_T(s_1)$ and $\mu_T(s_2)$ were even and odd, respectively. 

    Consider case (2) and assume w.l.o.g. that $e$ and $f$ cover segment $s_1$. Similar to the previous case, we know that any edge $g \in T$ where $g \neq e, f$, that covers $s_1$ must also cover $s_2$ and vice versa, as $v$ already has degree 2. Since there are $\mu_T(s_1)$ edges of $T$ that cover segment $s_1$, we have that $\mu_T(s_2) = 2 + \mu_T(s_1)$. Then, since we assumed $\mu_c(s_1)$ was even, $\mu_T(s_2)$ would be even, which is a contradiction. \hfill$\qed$
 \end{proof}

\subsection{Proof of Lemma \ref{lemma:4}} \label{appendix.sec3.Lemma.7}

We begin by stating a useful observation for odd-depth tours along the same vein as Observation \ref{claim:1}.

\begin{observation} \label{claim:5}
Consider the uniform metric graph $\uniformCircleGeomGraph$ for some $n \leq 8$. Then (i) there is no pair of edge-disjoint tours when $n \leq 4$, and (ii) there is no pair of odd-depth edge-disjoint tours with total cost less than $\frac{16}{5}n$ when $n \in \{5, 6, 7, 8\}$. 
\end{observation}
\begin{proof}
    Note that $\frac{16}{5}n = 3n + \frac{n}{5} < 3n+2$ since $n \leq 8$. Hence, any pair $(T_1, T_2)$ of disjoint tours of total cost less than $\frac{16}{5}n$ must be of one of the following three types: 
    \begin{itemize}
        \item[$\bullet$] $T_1 \cup T_2$ contains $n$ edges of length 1 and $n$ edges of length 2. 
        \item[$\bullet$] $T_1 \cup T_2$ contains $n-1$ edges of length 1 and $n+1$ edges of length 2. 
        \item[$\bullet$] $T_1 \cup T_2$ contains $n$ edges of length 1, $n-1$ edges of length 2, and a single edge of length 3. 
    \end{itemize}
    We have verified by computer that no such pairs of tours exist for $5 \leq n \leq 8$, and that no pair of disjoint tours exists when $n \leq 4$.\footnote{Code available at \hyperlink{https://github.com/andoresu47/Disjoint-Tours-and-the-PoD}{https://github.com/andoresu47/Disjoint-Tours-and-the-PoD}.}\hfill$\blacksquare$
\end{proof}

Next, notice that the extended notions of \textit{cut-point} and \textit{$\ell$-piece} from Hamiltonian paths to tours imply direct analogs of Claims \ref{claim:6} and \ref{claim:7} for odd-depth tours. To avoid redundancy, we omit the explicit statements of these tour-specific claims and instead use them to establish the following result, which mirrors Claim~\ref{lemma.1} from the path setting.

\begin{claim} \label{lemma.1.1}
    Any pair of disjoint tours on the uniform metric graph $\uniformCircleGeomGraph$ has a cut-point if there is a 3-piece with a total depth of less than 10. 
\end{claim}
\begin{proof}
    The proof of Claim~\ref{lemma.1} carries over directly by substituting Claims~\ref{claim:6} and~\ref{claim:7} with their analogs for odd-depth tours, and replacing Observation~\ref{fact:2} with the definition of an odd-depth tour. With these adjustments, the argument proceeds unchanged.\hfill$\blacksquare$
\end{proof}

The proof of Lemma \ref{lemma:4} follows the same structure as that of Lemma~\ref{lemma.2}, but uses Observation~\ref{claim:5} instead of Observation~\ref{claim:1}, and Claim~\ref{lemma.1.1} in place of Claim~\ref{lemma.1}. \hfill$\qed$

\subsection{Detailed description of Algorithm \texttt{Tours}} \label{appendix.algoToursDescription}
\begin{algorithm}[H]
\caption[Caption for LOF]{Algorithm \texttt{Tours}} \label{algo.approxTours}
\vspace{.5em}
\textbf{Input:} Uniform metric graph $G = (V, E)$ in \realCircle{}. \\
\textbf{Output:} A pair $(T_1, T_2)$ of disjoint tours. 
\vspace{.5em}
\begin{algorithmic}[1]
\newcommand\NoDo{\renewcommand\algorithmicdo{}}
\newcommand\ReDo{\renewcommand\algorithmicdo{\textbf{do}}}
\algrenewcommand\algorithmiccomment[1]{\hfill {\color{blue} \(\triangleright\) #1}}

\algdef{SE}[SUBALG]{Indent}{EndIndent}{}{\algorithmicend\ }%
\algtext*{Indent}
\algtext*{EndIndent}

\State Let $G'$ be a uniform metric graph in \realLine{} on the same set of vertices as $G$ plus an additional vertex $t$. 
\State Let $s$ and $t$ be the endpoints of the path graph of size $m = n+1$ formed by the vertex set and the segments of $G'$. 
\State With Algorithm \ref{algo.approxPaths}, find a pair $(H_1, H_2)$ of disjoint Hamiltonian ($s$, $t$)-paths in $G'$. 
\State Construct tour $T_1$ by contracting the vertex $t$ of $H_1$ into vertex $s$.
\State Construct tour $T_2$ by contracting the vertex $t$ of $H_2$ into vertex $s$.
\State Return $(T_1, T_2)$.
\end{algorithmic}
\vspace{.5em}
\end{algorithm}

\section{Proofs of Section \ref{sec.3}} \label{appendix.sec4}
\subsection{Proof of Claim \ref{claim:2}} \label{appendix.sec4.Claim.1}
\begin{proof}
We prove this only for segment $(v_2, v_3)$ since the argument applies symmetrically for $(v_{n-2}, v_{n-1})$. 
By Observation \ref{claim:0}, the first segment $s_1 = (v_1, v_2)$ of $P_n$ must have depths $\mu_{H_1}(s_1) = \mu_{H_2}(s_1) = 1$. This implies that at least one of the two paths $H_1$ or $H_2$ covers vertex $v_2$. But $v_2$ must have a degree of 2 in both paths. Hence, the second segment $s_2 = (v_2, v_3)$ has a depth $\mu_{H_1}(s_2) > 1$ or $\mu_{H_2}(s_2) > 1$. Moreover, by part (i) of Observation \ref{claim:1}, we have $\mu_{H_1}(s_2) \geq 3$ or $\mu_{H_2}(s_2) \geq 3$. \hfill$\qed$
\end{proof}

\subsection{Proof of Lemma \ref{lemma:poDofTSP2}} \label{appendix.proofOfLemmaPoDTSP2}

To prove the lemma, we first establish five structural properties concerning the depth of segments with respect to a pair of disjoint tours in $\mathbb{S}^1$. Let $G$ be a (not necessarily uniform) geometric graph in $\mathbb{S}^1$. Consider an arbitrary sequence $(s_1, s_2, s_3, s_4, s_5)$ of five consecutive segments of $G$. We begin with two claims concerning even-depth tours. 

\begin{claim} \label{observation.new1}
    Let $T$ be an even-depth tour in $G$. If $\mu_{T}(s_3) = 0$ and $s_2$ (resp. $s_4$) is not an edge of $T$, then $\mu_T(s_1)$ (resp. $\mu_T(s_5)$) is at least 4. 
\end{claim}
\begin{proof}
    We prove the case where $s_2 \notin T$; the argument for $s_4 \notin T$ is symmetric. Let $s_1 = (v_i, v_{i+1})$, $s_2 = (v_{i+1}, v_{i+2})$, and $s_3 = (v_{i+2}, v_{i+3})$. Since $\mu_T(s_3) = 0$, the two edges $e_1, e_2 \in T$ incident to $v_{i+2}$ cannot cover $s_3$, so both must cover $s_1$. Moreover, because $s_2 \notin T$, the two edges $e_3, e_4 \in T$ incident to $v_{i+1}$ must be distinct from $e_1$ and $e_2$. Additionally, these edges also cannot cover $s_2$ or $s_3$ (since $\mu_T(s_3) = 0$), so they too must cover $s_1$. Hence, $s_1$ is covered by at least four edges of $T$, and $\mu_T(s_1) \geq 4$. \hfill$\qed$
\end{proof}

The following claim follows directly from Claim \ref{observation.new1}, noting that in any pair $(T_1, T_2)$ of edge-disjoint tours such that $\mu_{T_1}(s_3) = \mu_{T_2} = 0$, segment $s_2$ (resp. $s_4$) can appear as an edge in at most one tour; see Figure \ref{fig:aFivePiece} for an illustration. 

\begin{figure}[!ht]
    \centering
    \begin{tikzpicture}
\node[draw=black,circle,inner sep=0pt,minimum size=1.8em] (v0) at (-2.5,2) {\footnotesize $v_i$};
\node[draw=black,circle,inner sep=0pt,minimum size=1.8em] (v1) at (-1.5,2) {\footnotesize $v_{i+1}$};
\node[draw=black,circle,inner sep=0pt,minimum size=1.8em] (v2) at (-0.5,2) {\footnotesize $v_{i+2}$};
\node[draw=black,circle,inner sep=0pt,minimum size=1.8em] (v3) at (0.5,2) {\footnotesize $v_{i+3}$};
\node[draw=black,circle,inner sep=0pt,minimum size=1.8em] (v4) at (1.5,2) {\footnotesize $v_{i+4}$};
\node[draw=black,circle,inner sep=0pt,minimum size=1.8em] (v5) at (2.5,2) {\footnotesize $v_{i+5}$};

\node (t1) at (-3.8,2) {\footnotesize $T_1: $};

\node[draw=black,circle,inner sep=0pt,minimum size=1.8em] (v0) at (-2.5,0.3) {\footnotesize $v_i$};
\node[draw=black,circle,inner sep=0pt,minimum size=1.8em] (v1) at (-1.5,0.3) {\footnotesize $v_{i+1}$};
\node[draw=black,circle,inner sep=0pt,minimum size=1.8em] (v2) at (-0.5,0.3) {\footnotesize $v_{i+2}$};
\node[draw=black,circle,inner sep=0pt,minimum size=1.8em] (v3) at (0.5,0.3) {\footnotesize $v_{i+3}$};
\node[draw=black,circle,inner sep=0pt,minimum size=1.8em] (v4) at (1.5,0.3) {\footnotesize $v_{i+4}$};
\node[draw=black,circle,inner sep=0pt,minimum size=1.8em] (v5) at (2.5,0.3) {\footnotesize $v_{i+5}$};

\node (t1) at (-3.8,0.3) {\footnotesize $T_2: $};

\draw (0.5,1.4) -- (2.5,1.4);
\draw (-0.5,1.5) -- (-2.5,1.5);
\draw (-0.5,1.4) -- (-2.5,1.4);
\draw (-1.5,2.5) -- (-2.5,2.5);
\draw (-1.5,2.6) -- (-2.5,2.6);
\draw (1.5,2.5) -- (2.5,2.5);

\draw (0.5,1.5) -- (1.5,1.5);

\draw (-2.5,-0.3) -- (-0.5,-0.3);
\draw (2.5,-0.2) -- (0.5,-0.2);
\draw (2.5,-0.3) -- (0.5,-0.3);
\draw (-1.5,0.8) -- (-2.5,0.8);

\draw (1.5,0.8) -- (2.5,0.8);
\draw (1.5,0.9) -- (2.5,0.9);
\draw (-1.5,-0.2) -- (-0.5,-0.2);
\end{tikzpicture}
    \caption{A sequence $(s_1, s_2, s_3, s_4, s_5)$ of five consecutive segments, where the middle segment $s_3=(v_{i+2}, v_{i+3})$ has depth 0 w.r.t. two edge-disjoint tours $T_1$ and $T_2$. For simplicity, tour edges are omitted; in their place, intervals are drawn to denote the vertices and segments covered by the tours.}
    \label{fig:aFivePiece}
\end{figure}

\begin{claim} \label{corollary:1.new}
    In any pair $(T_1, T_2)$ of disjoint even-depth tours such that $\mu_{T_1}(s_3) = \mu_{T_2}(s_3) = 0$, the depth of $s_1$ and $s_5$ w.r.t. one of the tours is at least 4. 
\end{claim} 

Claim \ref{corollary:1.new} holds for the case of two disjoint even-depth tours. Consider now a pair of disjoint tours, where one is an even-depth tour and the other is an odd-depth tour. We make the following claim. 

\begin{claim} \label{claim.new}
    Let $(T_1, T_2)$ be a pair of disjoint tours in $G$ such that $T_1$ is an even-depth tour and $T_2$ is an odd-depth tour. If $\mu_{T_1}(s_3) = 0$, then at least one of the following holds:
    \begin{enumerate}
        \item $\mu_{T_1}(s_3) + \mu_{T_2}(s_3) \geq 3$, or
        \item $\mu_{T_1}(s_1) + \mu_{T_2}(s_1) \geq 5$ and $\mu_{T_1}(s_5) + \mu_{T_2}(s_5) \geq 5$.
    \end{enumerate}
\end{claim}
%
%
\begin{proof}
    Since $T_2$ is an odd-depth tour, there are two cases for the depth of segment $s_3$ w.r.t. $T_2$: (i) $\mu_{T_2}(s_3) = 1$, or (ii) $\mu_{T_2}(s_3) \geq 3$. The first part of the claim follows directly from assuming (ii), so we focus on proving the second part by assuming (i). We distinguish the following cases: 
    \begin{description}
        \item[Case 1: $s_2 \not\in T_1 \textnormal{ and } s_4 \not\in T_1$.] Follows directly from Claim \ref{observation.new1} and the trivial bounds $\mu_{T_2}(s_1) \geq 1$ and $\mu_{T_2}(s_5) \geq 1$. 
        \item[Case 2: $s_2 \not\in T_1 \textnormal{ or } s_4 \not\in T_1$.] Assume w.l.o.g. that $s_2 \not\in T_1$. 
        By Claim \ref{observation.new1}, we have that $\mu_{T_1}(s_1) \geq 4$. On the other hand, because edge $s_4$ cannot be used by tour $T_2$, at least one edge of $v_{i+3}$ in $T_2$ must cover segment $s_5$. Moreover, because $\mu_{T_2}(s_3) = 1$, segment $s_5$ must also be covered by both of the edges of $v_{i+4}$ in $T_2$. Hence, $\mu_{T_2}(s_5) \geq 3$. Part 2 of the claim then follows by noting that $\mu_{T_2}(s_1) \geq 1$ and $\mu_{T_1}(s_5) \geq 2$. 
        \item[Case 3: $s_2 \in T_1 \textnormal{ and } s_4 \in T_1$.] Since $\mu_{T_2}(s_3) = 1$, both of the edges of $v_{i+4}$ (resp. $v_{i+1}$) must cross segment $s_5$ (resp. $s_1$). Moreover, since $s_4$ (resp. $s_2$) is inaccessible as an edge to $T_2$, at least one edge of $v_{i+3}$ (resp. $v_{i+1}$) must cross segment $s_5$ (resp. $s_1$). Therefore, $\mu_{T_2}(s_5) \geq 3$ (resp. $\mu_{T_2}(s_1) \geq 3$). Since $\mu_{T_1}(s_5) \geq 2$ and $\mu_{T_1}(s_1) \geq 2$, part 2 of the claim follows. 
    \end{description}
    Since Cases 1-3 cover all possibilities for the inclusion of segments as edges of $T_1$, the claim is proven. \hfill$\qed$
    \end{proof}
    
The next two claims concern only odd-depth tours. Let $C_n$ denote the cycle graph formed by the vertex set and segments of $G$. For a given pair $(T_1, T_2)$ of disjoint odd-depth tours in $G$, we say that a sequence $\sigma = (s_1, s_2, \ldots, s_d)$ of consecutive segments $s_i \in C$ is a  \textit{1-section w.r.t. $T_1$ and $T_2$} if $\mu_{T_1}(s_i) = \mu_{T_2}(s_i) = 1$ for all $i \in \{1, \ldots, d\}$. The set of all 1-sections is denoted by $\Sigma(T_1, T_2)$. 

\begin{claim} \label{claim:3}
    Any 1-section $\sigma \in \Sigma(T_1, T_2)$ has a length of at most 2.
\end{claim}
\begin{proof}
     For the sake of contradiction, suppose that $\sigma$ has length 3 or more. Let $(s_1, s_2, s_3)$ denote an arbitrary subsequence of length 3 (i.e., a 3-piece) in $\sigma$, where $s_1 = (v_1, v_2)$, $s_2 = (v_2, v_3)$, and $s_3 = (v_3, v_4)$. Consider such 3-piece w.r.t. tours $T_1$ and $T_2$. We know that the vertices $v_2$ or $v_3$ cannot be covered by either tour, since that would cause segments $s_1$ or $s_3$ to have depth greater than 1 w.r.t. at least one of $T_1$ and $T_2$. Therefore, for the segment $s_2$ to have a total depth of 2, the edge $(v_2, v_3)$ must be present in both tours $T_1$ and $T_2$. But the tours are supposed to be edge-disjoint. Hence, such a 3-piece cannot exist, which gives the required contradiction. In consequence, any 1-section must have a length of at most 2. \hfill$\qed$
\end{proof}

In other words, Claim \ref{claim:3} states that at most two consecutive segments in $C_n$ have each depth 1 w.r.t. both $T_1$ and $T_2$. Now, let $\sigma_1, \sigma_2 \in \Sigma(T_1, T_2)$ be an arbitrary pair of 1-sections. We define the distance between $\sigma_1$ and $\sigma_2$ to be the minimum number of segments separating a segment in $\sigma_1$ from a segment in $\sigma_2$. 

\begin{claim} \label{claim:4}
    The distance between any two distinct 1-sections $\sigma_1, \sigma_2 \in \Sigma(T_1, T_2)$ is at least 2. 
\end{claim}
\begin{proof}
    For the sake of contradiction, suppose that the distance between $\sigma_1$ and $\sigma_2$ is exactly 1. (Notice that $\sigma_1$ and $\sigma_2$ having distance 0 is the same as having a longer 1-section resulting from the concatenation of $\sigma_1$ with $\sigma_2$.) That is, there exist three consecutive segments $f, g, h \in C_n$, with $f = (v_1, v_2)$, $g = (v_2, v_3)$, and $h = (v_3, v_4)$, such that $\mu_{T_1}(f) = \mu_{T_2}(f) = 1$, $\mu_{T_1}(h) = \mu_{T_2}(h) = 1$, and $\max(\mu_{T_1}(g), \mu_{T_2}(g)) > 1$. W.l.o.g., suppose that $\mu_{T_1}(g) > 1$. Hence, at least one of $v_2$ or $v_3$ is covered by $T_1$. But this would imply that one of the segments $f$ or $h$ is also covered by $T_1$, meaning that $\mu_{T_1}(f) > 1$, or $\mu_{T_1}(h) > 1$. This violates the condition that all segments in $\sigma_1$ and $\sigma_2$ have depth 1 w.r.t. both tours $T_1$ and $T_2$. Hence, we obtain the required contradiction, and the claim is proven. \hfill$\qed$
\end{proof}

With Claims \ref{corollary:1.new}-\ref{claim:4}, we have the necessary tools to provide a lower bound for the PoD of \disjointTSP{} in a general metric, thus proving Lemma \ref{lemma:poDofTSP2}. 

\poDofTSPTwo*
\begin{proof}
    Consider a metric graph $G = (V, E)$ in $\realCircle$ of size $n > 6$, and let $C = (V, E')$ be the cycle graph whose edge set is formed by the segments of $G$. Let $f, g, h \in E$ be any triple of consecutive edges of $C$. The edge weights of $G$ are defined as follows: $w(e) = 1$ for all $e \in E \setminus\{f, h\}$, and $w(f) = w(h) = W$. 

    Now, let $T_1$ and $T_2$ be a pair of edge-disjoint tours in $G$. There are four possibilities depending on whether the tours are odd- or even-depth tours. First, we consider the case when both $T_1$ and $T_2$ are odd-depth tours. With the aid of Claims \ref{claim:3} and \ref{claim:4}, we show that the depth of at least one of the segments $f$ or $h$ is at least 3 w.r.t. one of the two tours $T_1$ or $T_2$, meaning that the cost of one of the tours is at least $4W$.  
    
    For the sake of contradiction, suppose that segments $f$ and $h$ both satisfy that $\mu_{T_1}(f) = \mu_{T_2}(f) = 1$ and $\mu_{T_1}(h) = \mu_{T_2}(h) = 1$. There are two cases for the depth of segment $g$ w.r.t. $T_1$ and $T_2$: (i) at least one of $\mu_{T_1}(g)$ or $\mu_{T_2}(g)$ is greater than 1, or (ii) $\mu_{T_1}(g) = \mu_{T_2}(g) = 1$. Consider Case (i) first. W.l.o.g., assume that $\mu_{T_1}(g) > 1$. Then, the distance between the 1-sections $\sigma_1$ and $\sigma_2$ such that $f \in \sigma_1$ and $h \in \sigma_2$, is 1. But, by Claim \ref{claim:4}, we know that such distance must be 2 or more. Hence, we have a contradiction. Now, we turn to Case (ii). If $\mu_{T_1}(g) = \mu_{T_2}(g) = 1$, the sequence $\sigma = (f, g, h)$ of contiguous segments of $C$ would be an 1-section of length 3. But, by Claim \ref{claim:3}, we know that to be impossible. Hence, we get the required second contradiction. Having proved Cases (i) and (ii), we conclude that at least one of $\mu_{T_1}(f)$, $\mu_{T_2}(f)$, $\mu_{T_1}(h)$, or $\mu_{T_2}(h)$ is greater than 1. W.l.o.g., assume that $\mu_{T_1}(f) > 1$. Further, because of Lemma \ref{lemma:3}, it must be that $\mu_{T_2}(f) \geq 3$. 
    
    Now, since we are proving a lower bound, we consider the best case scenario when $\mu_{T_1}(h) = \mu_{T_2}(h) = 1$. Then, the cost of $T_1$ is $c(T_1) \geq 4 \cdot W + (n - 2)$. On the other hand, the cost of an optimal tour is $OPT = 2 \cdot W + (n - 2)$. Then, we have: 
    \begin{equation}
        PoD(\disjointSHP{}) \geq \frac{4W+ (n - 2)}{2W + (n - 2)} = 2 - \frac{n - 2}{2W + (n - 2)}. \label{eq.poD.TSP}
    \end{equation}
    Thus, for any $\varepsilon > 0$, we can obtain a PoD of $2 - \varepsilon$ by setting $W:= \frac{(1 - \varepsilon)(n-2)}{2\varepsilon}$. 

    Consider now the case when at least one of the tours is an even-depth tour, which we assume w.l.o.g. to be $T_1$. There are two subcases. First, if $T_1$ covers both segments $f$ and $h$, then $c(T_1) \geq 4 \cdot W + (n - 2)$, while $OPT = 2 \cdot W + (n - 2)$. Hence, we obtain again inequality \eqref{eq.poD.TSP}. The second subcase is when $\mu_{T_1}(f) = 0$ or $\mu_{T_1}(h) = 0$. (Note that both cannot happen simultaneously, as $T_1$ would not be a tour.) W.l.o.g., assume that $\mu_{T_1}(f) = 0$, and consider the sequence of five consecutive segments $(s_1, s_2, s_3, s_4, s_5)$ such that $s_3 = f$, $s_4 = g$, and $s_5 = h$. There are two further cases depending on whether $T_2$ is an even- or odd-depth tour. 
    
    If $T_2$ is an even-depth tour, the only case of interest is when $\mu_{T_2}(f) = 0$ (otherwise, $T_2$ covers both segments $f$ and $h$ twice and $c(T_2) \geq 4 \cdot W + (n - 2)$). By Claim \ref{corollary:1.new}, the depth of segment $h$ w.r.t. one of $T_1$ or $T_2$ is 4; which we assume w.l.o.g. to be $T_2$. Then $c(T_2) \geq 4 \cdot W + (n - 2)$ and we recover inequality \eqref{eq.poD.TSP}. The only case remaining is when $T_2$ is an odd-depth tour, where it follows from Claim \ref{claim.new} that at least one of $T_1$ or $T_2$ has a cost that is at least $4 \cdot W + (n - 2)$, again yielding the lower bound from inequality \eqref{eq.poD.TSP}.

    Since we have covered all possible combinations of tour types (i.e., odd or even), and in each case shown that the ratio in \eqref{eq.poD.TSP} holds, the lemma is proven. \hfill$\qed$
\end{proof}

\subsection{Proof of Lemma \ref{lemma:6}} \label{appendix.sec4.Lemma.12}
\begin{proof}
    If $n$ is odd, then $c(T_1) = c(H)$ because $T_1$ covers each segment of $G$---an edge of $H_\mathrm{opt}$---exactly once. On the other hand, since $T_2$ covers each segment of $G$ exactly twice, we have $c(T_2) = 2 \cdot c(H)$. Thus, $cost(T_1, T_2) / c(H_\mathrm{opt}) = 2$. 
    
    For the case where $n$ is even, $T_1$ covers $n-2$ segments exactly once, and one segment $s$ exactly thrice, giving $c(T_1) = c(H_\mathrm{opt}) + 2 \cdot c_s$, where $c_s$ is the cost of segment $s$. Without loss of generality, we may assume that $s$ is the segment with the lowest cost in $H_\mathrm{opt}$. (Otherwise, simply shift both cycles $T_1$ and $T_2$ until segment $s$ indeed achieves minimum cost.) Hence, we have $n \cdot c_s \leq c(H_\mathrm{opt})$. 
    
    On the other hand, $T_2$ covers every segment except for $s$ exactly twice, resulting in $c(T_2) = 2 \cdot(c(H_\mathrm{opt}) - c_s)$. Because we assumed that $c_s$ is of minimum cost, and there is no pair of edge-disjoint tours for $n < 4$, it must be that $cost(T_1, T_2) = c(T_2)$. Therefore, $cost(T_1, T_2) / c(H_\mathrm{opt}) < 2$ and the lemma is proven.  \hfill$\qed$
\end{proof}

\end{document}